\documentclass{article}

\usepackage[main, final, preprint]{neurips_2026}
\usepackage[utf8]{inputenc} % allow utf-8 input
\usepackage[T1]{fontenc}    % use 8-bit T1 fonts
\usepackage{hyperref}       % hyperlinks
\usepackage{url}            % simple URL typesetting
\usepackage{booktabs}       % professional-quality tables
\usepackage{amsfonts}       % blackboard math symbols
\usepackage{nicefrac}       % compact symbols for 1/2, etc.
\usepackage{microtype}      % microtypography
\usepackage{xcolor}         % colors
\usepackage{wrapfig}

\usepackage[most]{tcolorbox}

\newtcolorbox{researchquestion}{
  colback=gray!5,
  colframe=gray!60,
  boxrule=0.6pt,
  arc=1mm,
  left=6pt,
  right=6pt,
  top=5pt,
  bottom=5pt
}

\title{Characterizing the Trainability of Instantaneous Quantum Polynomial Circuit Born Machines}

\author{%
  Kevin Shen\thanks{\texttt{k.shen@liacs.leidenuniv.nl}} \\
  $\langle aQa \rangle ^L$, Leiden University\\
  LIACS, Leiden University\\
  BMW Group\\
  \And
  Susanne Pielawa\\
  BMW Group\\
  \AND
  Vedran Dunjko\\
  $\langle aQa \rangle ^L$, Leiden University\\
  LIACS, Leiden University\\
  \And
  Hao Wang \thanks{\texttt{h.wang@liacs.leidenuniv.nl}}\\
  $\langle aQa \rangle ^L$, Leiden University\\
  LIACS, Leiden University\\
}

\usepackage{amsmath,amsfonts,bm,}

\newcommand{\vttheta}{{\bm{\theta}}}
\newcommand{\vta}{{\bm{a}}}

\newcommand{\vtb}{{\bm{b}}}

\newcommand{\vtx}{{\bm{x}}}

\newcommand{\vtsj}{{\bm{s}^j}}

\newcommand{\vtsl}{{\bm{s}^\ell}}
\newcommand{\vts}{{\bm{s}}}
\usepackage{physics}

\def\eqref#1{equation~\ref{#1}}
\def\1{\bm{1}}

\DeclareMathAlphabet{\mathsfit}{\encodingdefault}{\sfdefault}{m}{sl}
\SetMathAlphabet{\mathsfit}{bold}{\encodingdefault}{\sfdefault}{bx}{n}

\newcommand{\E}{\mathbb{E}}

\usepackage{amsmath}
\usepackage{amssymb}
\usepackage{mathtools}
\usepackage{amsthm}
\usepackage{xspace}

\usepackage[capitalize]{cleveref}
\usepackage{subcaption}

\theoremstyle{plain}
\newtheorem{theorem}{Theorem}[section]
\newtheorem{proposition}[theorem]{Proposition}
\newtheorem{lemma}[theorem]{Lemma}
\newtheorem{example}[theorem]{Example}
\newtheorem{corollary}[theorem]{Corollary}
\theoremstyle{definition}
\newtheorem{definition}[theorem]{Definition}

\theoremstyle{remark}
\newtheorem{remark}[theorem]{Remark}

\usepackage[textsize=tiny]{todonotes}

\begin{document}

\maketitle

\begin{abstract}
Instantaneous Quantum Polynomial Circuit Born Machines (IQP-QCBMs) have been proposed as quantum generative models that combine a classically tractable training objective—based on the maximum mean discrepancy (MMD)—with a potential quantum advantage motivated by sampling-complexity arguments. While recent works have explored this model across various application domains, fundamental questions remain: does the model suffer from exponentially vanishing loss gradients, known as the barren plateau problem—a pervasive obstacle in quantum machine learning—and how do regimes of trainability relate to regimes of possible quantum advantage? Here, we address both questions analytically. To study trainability, we derive closed-form expressions for the variances of the partial derivatives of the MMD loss function and establish general upper and lower bounds. We explicitly characterize how trainability depends on the generator set and the spectrum of the chosen kernel, identifying regimes in which low Hamming-weight kernels avoid exponential gradient suppression under structured topologies. Regarding potential quantum advantage, we reformulate the anti-concentration property in terms of the same generator-set quantities that govern trainability, enabling a unified analysis. We show that sparse IQP architectures can produce classically intractable output distributions while simultaneously remaining trainable, at least at lower-weight frequencies. Our analytical results corroborate the numerical observations of prior work and provide principled guidelines for designing scalable IQP-QCBM architectures.
\end{abstract}

\section{Introduction}
Instantaneous Quantum Polynomial Circuit Born Machines (IQP-QCBM)~\citep{shepherd09,Liu18} are emerging quantum generative models that have attracted increasing attention. IQP-QCBMs represent probability distributions over bitstrings through Born-rule measurements of parameterized IQP circuits. IQP-QCBMs are distinguished from many other quantum machine learning (QML) models in two ways: First, the models are conjectured to inherit the classical sampling hardness from IQP circuits under standard complexity-theoretic assumptions~\citep{Bremner10,Bremner16}. Second, as \citet{recioarmengol25} showed, by exploiting the special circuit architectures of IQP circuits, the models can be trained on a classical computer using the Maximum Mean Discrepancy (MMD) loss~\citep{gretton12}, while deployment on a quantum computer is required for sample generation. IQP-QCBMs have been experimentally explored and trained on high-dimensional datasets across multiple application domains such as life sciences~\citep{11250278}, fluid dynamics~\citep{liu2026generativequantumutilitycorrelationcomplexity}, and graphs~\citep{jcdk-q3xc}. 

Despite the growing body of exciting application-oriented work on IQP-QCBMs, their scalability depends on a central unresolved question: trainability. 
Here, trainability refers to whether a practical optimization procedure, such as stochastic gradient descent, can reduce the loss within a reasonable time. 
A major obstruction to trainability in QML is the barren plateau (BP) phenomenon, in which the variance of the loss gradient with respect to the parameter distribution decays exponentially with the number of qubits $n$~\citep{mcclean18,Cerezo21,Holmes22,Larocca25}. In the presence of BP, the exponentially small gradient information will make both gradient- and sampling-based optimization ineffective for large $n$. To our knowledge, the occurrence of BPs in IQP-QCBMs has not been formally studied. The fact that the MMD loss and its gradients for IQP-QCBMs can be evaluated classically does not imply trainability; it only addresses the cost of computing the loss signal, rather than this signal being significant enough for effective training.

In this work, we show that whether an IQP-QCBM exhibits a BP depends on two practical design choices: the IQP circuit architecture (see~\cref{def:IQP}) and the kernel used in the MMD loss (see~\cref{prop:mmd-loss}). We derive a closed-form expression for the variance of the MMD loss gradient, and we identify the "critical rank" of the IQP-QCBM model, the dimension of the space spanned by the anti-commuting generators with respect to the observable, as the deciding factor for the presence or absence of BPs. Also, using four example circuit architectures, we demonstrate the scaling of their loss gradients and show that, in many cases, BPs can be mitigated with an appropriately chosen kernel. Our analytical work provides guidance for avoiding architectures that are provably susceptible to BPs and narrows the design space of trainable IQP-QCBMs.

\begin{table}
\centering
\caption{Summary of the main result. We consider four architectures for the IQP circuit (\cref{pic:illustration}) and report the scaling of the critical rank $r^\vta$ (\cref{thm:uniform})---the key circuit property determining the scaling of the MMD partial derivatives. ``No B.P.'' indicates the absence of a barren plateau in the MMD loss with a suitable kernel choice, which is necessary for trainability at scale. ``Anti-concentration'' means $L^2$ anti-concentration (\cref{def:anti-concentration}) of the IQP-QCBM's output distribution, which is needed to claim classical non-simulability.}
\vspace{0.2cm}
\label{tab:examples-summary}
\begin{tabular}{cccc}
\toprule
IQP Architecture 
& $\abs{\vta}$ s.t. $r^\vta = \mathcal{O}(\log n)$
& No B.P.
& Anti-concentration \\
\midrule
Product State
& $\mathcal{O}(\log n)$
& $\checkmark$
& $\times$ \\

2D Lattice
& $\mathcal{O}(\log n)$
& $\checkmark$
& $\times$ \\

Sparse ER Graph ($p>2\log n /n$)
& $\mathcal{O}(1)$
& $\checkmark$
& $\checkmark$ \\

Complete Graph
& N.A.
& $\times$
& $\checkmark$ \\
\bottomrule
\end{tabular}
\vspace{-0.5cm}
\end{table}

Another important criterion for QML architectures is that they possess representational or computational power beyond efficient classical simulation, thereby providing the possibility of quantum advantage. Accordingly, promising IQP-QCBM architectures must be assessed not only for trainability, but also for whether their output distribution families remain classically hard to simulate. Existing classical sampling-hardness results for IQP circuits have been established only under specific architectural assumptions~\citep{Bremner10,Bremner16,Bremner17}. Indeed, certain restricted IQP architectures are provably classically simulable; for example, planar IQP circuits are connected to exactly solvable Ising models and can be efficiently simulated in relevant regimes~\citep{Fujii17}.

This issue is further highlighted by recent work on supervised QML, which argues that many architectures with provable absence of barren plateaus are also efficiently classically simulable, because the mechanisms that avoid gradient concentration often restrict the computation to polynomially-sized subspaces~\citep{Cerezo2025Simulability}.
It is therefore natural to ask whether classical non-simulability and trainability can coexist in a quantum generative model, which is a prerequisite for a potentially useful quantum advantage.

In this work, we provide a positive answer to this question by showing that a high-probability family of sparse Erd\H{o}s--R\'enyi architectures with $\mathcal{O}(n\log n)$ edges has inverse-polynomial low-order moment-gradient variances, while no uniform efficient classical sampler can approximate every instance in this family, assuming the sparse average-case hardness conjecture and noncollapse of the polynomial hierarchy~\citep{Bremner17}. We summarize the application of our main results regarding critical rank, BPs, and anti-concentration to the example architectures in~\cref{tab:examples-summary}.

\section{Background and related work} \label{sec:background}
\paragraph{Notations} We first introduce some notations: $\mathbb{F}_2^n$ denotes the $n$-dimensional vector space over the binary field $\{0,1\}$; $\rho = |+\rangle\!\langle+|^{\otimes n}$ is the initial quantum state; $\{\vtsj\}_j$ is a set of bitstrings in $\mathbb{F}_2^n$ and $\bm{s}_k^j\in\{0,1\}$ is the $k$-th component of $\vtsj$; $Z^{\vtsj} = \bigotimes_{k=1}^n Z^{\bm{s}_k^j}$ is a Pauli-$Z$ string operator, which acts nontrivially on the qubits indexed by the nonzero entries in $\vtsj$; $X^\vta = \bigotimes_{k=1}^n X^{a_k}$ is a Pauli-$X$ string observable, which, similarly, acts nontrivially on the qubits indexed by the nonzero entries in $\vta$.

For a probability distribution $p$ over $\mathbb{F}_2^n$, we denote its characteristic function by $C_p(\vta) = \mathbb{E}_{\vtx \sim p}[(-1)^{\vtx\cdot \vta}]$. We also consider a stationary, bounded kernel function $k(\vta, \vtb) = \phi(\vta\oplus\vtb)$ for some positive-definite function $\phi\colon \mathbb{F}_2^n \to \mathbb{R}$, and its Fourier transform $\Lambda(\vta) = 2^{-n}\sum_{\vtx\in\mathbb{F}_2^n} (-1)^{\vtx \cdot \vta} \phi(\vtx) $, known as the spectral measure of $k$.

\begin{definition}[IQP-QCBM]
\label{def:IQP}
An $n$-qubit IQP-QCBM is a variational quantum generative model defined by a commuting diagonal circuit
$U_{\vttheta}
=
\prod_{j=1}^D \exp\!\left(i\theta_j Z^{\vts^j}\right),
\theta_j \in [0,2\pi), \vts^j \in \mathbb F_2^n,$
where $D \in \mathcal O(\mathrm{poly}(n))$. The model prepares
$U_{\vttheta}\ket{+}^{\otimes n}$ and measures all qubits in the Pauli-$X$ basis.
By the Born rule, sampling from the model follows the distribution: For each $\vtx \in \mathbb{F}_2^n$, $
q_{\vttheta}(\vtx)
=
\operatorname{Tr}
\left[
|\vtx\rangle\!\langle\vtx|
H^{\otimes n}
U_{\vttheta}\rho U_{\vttheta}^{\dagger}
H^{\otimes n}
\right]$, where $\rho = |+\rangle\!\langle+|^{\otimes n}$.
We refer to $S=\{\vtsj\}_j$ as the generator set, which determines the model architecture.
\end{definition}

Previous work~\citet{recioarmengol25} proposed training IQP-QCBMs using the MMD loss. 

\begin{definition}[Maximum Mean Discrepancy (MMD)]
\label{def:mmd-kernel}
Let $p$ and $q$ be probability distributions on a data space $\mathcal{X}$, and let $k\colon\mathcal{X}\times\mathcal{X}\to\mathbb{R}$ be a positive-definite kernel. The squared maximum mean discrepancy between $p$ and $q$ is
\begin{equation}
\mathrm{MMD}_k^2(p,q)
=\mathbb{E}_{x,x'\sim p}\big[k(x,x')\big]
-2\,\mathbb{E}_{x\sim p,\,y\sim q}\big[k(x,y)\big]
+\mathbb{E}_{y,y'\sim q}\big[k(y,y')\big],
\end{equation}
where $x,x'$ are independent draws from $p$ and $y,y'$ are independent draws from $q$.
\end{definition}

For IQP-QCBMs, it has been shown that the MMD loss admits an equivalent expression in terms of characteristic functions. For the output distribution $q_\vttheta$, its characteristic function satisfies $C_\vttheta(\vta)=\operatorname*{\mathbb{E}}_{\vtx \sim q_\vttheta}[(-1)^{\vtx \cdot \vta}]=\Tr(X^{\vta}U_\vttheta\rho U_\vttheta^\dagger)$. These characteristic-function values are Pauli-$X$ string expectation values that can be estimated by Monte Carlo methods on classical computers for IQP circuits $U_\vttheta$.

\begin{proposition} [MMD loss in IQP-QCBM]\label{prop:mmd-loss}
Consider an IQP model and its output distribution $q_\vttheta$ defined in~\cref{def:IQP}. Given a target distribution $p$ over $\mathbb{F}_2^n$ and a kernel $k$, the MMD loss between $p$ and $q_\vttheta$ admits the following form~\citep{rudolph24, recioarmengol25}: $\mathcal{L}(\vttheta) \coloneqq \mathrm{MMD}^2_k(p, q_\vttheta) = \operatorname*{\mathbb{E}}_{\vta \sim \Lambda} \left[ \left( C_p(\vta) - C_\vttheta(\vta) \right)^2 \right]$,
where $C_p(\vta)  = \operatorname*{\mathbb{E}}_{\vtx \sim p}[(-1)^{\vtx \cdot \vta}]$ and $C_\vttheta(\vta)=\Tr(X^{\vta}U_\vttheta\rho U_\vttheta^\dagger)$ and $\Lambda$ is the Fourier transform of $k$.  
\end{proposition}
In QML, loss functions expressed in terms of expectation values often suffer from BPs, whereby gradients vanish exponentially with the number of qubits $n$~\citep{mcclean18}. BPs have been extensively studied, particularly in supervised learning settings. For instance, \citet{Cerezo21} characterized the dependence of gradient variance on circuit depth, while \citet{Wang21} showed that hardware noise can induce BPs. More recently, \citet{Ragone24} proposed a general algebraic framework for bounding gradient variance using the dynamical Lie algebra (DLA) generated by the circuit, and \citet{fontana2024characterizing} obtained related results using Lie group representation theory.

However, we note that the DLA-based approach of~\citet{Ragone24} cannot be directly applied to IQP-QCBMs. Concretely, applying the main result of ~\citet{Ragone24} requires that, first, the Lie algebra of the generators is semisimple, i.e., $\mathfrak{g} = \mathfrak{g}_{1}\oplus \mathfrak{g}_{2}\oplus\cdots$, where each $\mathfrak{g}_i$ is simple, and second, either the initial state $\rho$ or the observable $O$ has non-zero $\mathfrak{g}_{i}$-purity for at least one $i$. In the setting of IQP-QCBMs, the DLA is the abelian algebra $\mathfrak{g}=\operatorname{span}_{\mathbb{R}}\{iZ^{\vtsj}\}_j$, and the $\mathfrak{g}$-purity of both the observables $X^{\vta}$, for all $\vta\neq\bm 0$, and the initial state $\ket{+}\bra{+}$ is zero. Likewise, the framework of~\citet{fontana2024characterizing} applies to Lie-algebra-supported ansätze, i.e., $iO\in\mathfrak{g}$, whereas here $iX^{\vta}\notin\mathfrak{g}$. Therefore, in this paper, we derive the analytical formulas of gradient variances from scratch.

\begin{definition}[Barren Plateau] \label{def:bp}
Let $\Theta_n=[0,2\pi)^D$ be the parameter space of an $n$-qubit variational quantum
model with $D$ parameters, and let $\mu_n$ be a reference probability measure on $\Theta_n$.
We say that the loss landscape $\mathcal L:\Theta_n \to \mathbb{R}$ exhibits a \emph{barren plateau (BP) with respect to $\mu_n$} if, for every trainable parameter $\theta_\ell$, there exist constants $\alpha,\beta>0$, independent of $n$, such that
$$
\mathbb E_{\vttheta\sim\mu_n}
\left[
\left|
\frac{\partial \mathcal L(\vttheta)}{\partial \theta_\ell}
\right|^2
\right]
\le \alpha e^{-\beta n} 
$$
\end{definition}
As a consequence of~\cref{def:bp}, by Markov's inequality, for a loss landscape with B.P., \begin{equation}
    \forall \varepsilon>0,\; \mu_n\left(\left\{ \vttheta\in\Theta_n \colon \left| \frac{\partial \mathcal L(\vttheta)}{\partial \theta_\ell} \right| \ge \varepsilon \right\} \right) \le \frac{\alpha e^{-\beta n}}{\varepsilon^2},
\end{equation}
meaning that the partial derivatives are exponentially small with high probability
with respect to $\mu_n$.

\begin{remark}[Choice of the reference measure $\mu_n$]
We use the product-uniform measure, $\theta_j\sim\mathcal U[0,2\pi)$, i.e., normalized Haar measure on the parameter torus. This choice is required for the critical-rank formulas in \cref{thm:uniform}. The distinction between polynomial and exponential gradient scaling remains unchanged for measures whose densities relative to the uniform measure are bounded above and below by polynomial factors, while \cref{prop:lower} applies to any measure symmetric about zero.
\end{remark}

\section{Main results} \label{sec:main-result}
In this section, we present our main results. We begin in~\cref{subsec:mmd} by deriving a closed-form expression for the variance of the partial derivatives of the MMD loss, expressed in terms of the derivatives of the characteristic functions and the spectral measure of the kernel.
In~\cref{subsec:kernel-choice}, we analyze how the choice of kernel, equivalently the choice of spectral measure, influences MMD partial derivative variances.
In~\cref{subsec:uniform}, we introduce the critical rank $r^\vta$, a central architectural quantity for IQP-QCBMs, and show how it determines the presence or absence of BPs. We further apply the analysis to four representative architectures. 
Finally, in \cref{subsec:non-simulability}, we investigate the non-simulability of IQP-QCBM output distribution. 
We establish connections between the critical rank $r^\vta$ and the anti-concentration property, and identify a high-probability low-rank architecture family with low-moment trainability and conditional classical sampling hardness.

\subsection{Closed-form expression of MMD partial derivatives}
\label{subsec:mmd}
We can decompose the MMD partial derivative variance into frequency-wise contributions involving characteristic function values and their partial derivatives, as in \cref{eq:decomposition}.
\begin{proposition}[Variance of the MMD partial derivative]
\label{prop:lower}
Consider an $n$-qubit IQP-QCBM in~\cref{def:IQP}, with generator set $S=\{\vtsj\}_{j=1}^D$, and the MMD loss $\mathcal{L}$ defined in~\cref{prop:mmd-loss}. Let $\mu_n$ be a reference measure on the parameter space that is symmetric about 0 (e.g., uniform), and let $\Lambda$ be a positive spectral probability measure on $\mathbb{F}_2^n$. Then the variance of the partial derivative of the MMD admits the decomposition
\begin{align}
&\operatorname*{Var}_{\vttheta\sim \mu_n} \left[ \frac{\partial \mathcal{L}(\vttheta)}{\partial \theta_\ell} \right]
= 4 \operatorname*{\mathbb{E}}_{\vttheta\sim \mu_n} \left[ \sum_{\vta, \vtb \in \mathbb{F}_2^n} \Lambda(\vta) \Lambda(\vtb)F(\vta, \vtb)\right], \label{eq:decomposition} \\
&F(\vta, \vtb) = \left( C_\vttheta(\vta) C_\vttheta(\vtb) - 2C_p(\vta) C_\vttheta(\vtb) + C_p(\vta) C_p(\vtb) \right)\frac{\partial C_\vttheta(\vta)}{\partial \theta_\ell} \frac{\partial C_\vttheta(\vtb)}{\partial \theta_\ell} \nonumber
\end{align}
\end{proposition}
To isolate the dependence of the MMD gradient variance on the circuit architecture and kernel choice, we average this quantity over an unbiased ensemble of target distributions $\mathcal P$. Specifically, we assume that, $\forall\vta, \vtb \in \mathbb{F}_2^n \setminus \{\bm{0}\},\operatorname*{\mathbb {E}}_{p\sim\mathcal {P}}[C_p(\vta)]=0$, and $\operatorname*{\mathbb {E}}_{p\sim\mathcal {P}}\!\left[C_p(\vta)C_p(\vtb)\right]=\delta_{\vta,\vtb}\,\sigma^2(\vta)$, where $\sigma^2(\vta) = \operatorname*{\mathbb {E}}_{p \sim \mathcal{P}}\left[C_p(\vta)^2\right]$. Under this ensemble average, we obtain the lower bound:
\begin{equation}
\operatorname*{\mathbb{E}}_{p\sim\mathcal{P}} \left[\operatorname*{Var}_{\vttheta\sim\mu_n}  
\left[\frac{\partial \mathcal{L}(\vttheta)}{\partial \theta_\ell} \right]\right] \geq
4 \sum_{\vta\in \mathbb{F}_2^n} \Lambda^2(\vta)\sigma^2(\vta)\operatorname*{Var}_{\vttheta\sim\mu_n} \left[ \frac{\partial C_\vttheta(\vta)}{\partial \theta_\ell} \right].
\label{eq:lower-bound}
\end{equation}
We emphasize that \cref{eq:lower-bound} depends on squared spectral weights $\Lambda^2(\vta)$. This dependence is structural: the MMD gradient is a weighted sum over Fourier frequencies, and its variance therefore involves products of spectral weights. By contrast, the following loose upper bound depends only linearly on $\Lambda(\vta)$. This upper bound can be saturated only by highly structured target distributions.
\begin{corollary}[Upper bound on the variance of the MMD partial derivative] \label{prop:upper} 
Under the same conditions as required by~\cref{prop:lower}, the variance of the partial derivative of the MMD satisfies
\begin{equation}
    \operatorname*{Var}_{\vttheta \sim\mu_n} \left[ \frac{\partial \mathcal{L}(\vttheta)}{\partial \theta_\ell} \right] \leq 16 \sum_{\vta\in \mathbb{F}_2^n} \Lambda(\vta)\left[  \operatorname*{Var}_{\vttheta\sim\mu_n} \left[ \frac{\partial C_\vttheta(\vta)}{\partial \theta_\ell} \right]\right]
\end{equation}
\end{corollary}
\cref{prop:lower} and \cref{prop:upper} relate the variance of the MMD loss to the characteristic functions of IQP-QCBM output distribution. These bounds form the basis for the trainability analysis in the remainder of the paper.

\subsection{Impact of the kernel}
\label{subsec:kernel-choice}
The choice of kernel, or equivalently its spectral measure, can have a significant impact on the presence or absence of BPs. An observation from \cref{prop:lower} and \cref{prop:upper} suggests that if $\Lambda$ assigns only exponentially small mass to the frequencies whose characteristic-function gradients are non-negligible, then the MMD gradient may also become exponentially suppressed. 

In practice, one should therefore choose a kernel whose spectral measure places a sufficiently large mass on frequencies for which the characteristic-function gradients decay at most polynomially, whenever such frequencies exist. 

As we will see in the next section, characteristic-function gradients at high Hamming-weight frequencies typically decay faster than those at low-Hamming-weight frequencies. It is therefore natural to use kernels whose spectral measures concentrate on low-Hamming-weight frequencies. A standard example is the Gaussian kernel with bandwidth $\sigma=\Theta(n)$, whose spectral measure is $\Lambda_\sigma(\vta) = (1-\tau(\sigma))^{n-\abs{\vta}} \tau(\sigma)^\abs{\vta}$, where $\tau(\sigma) = [1-\exp(-1/(2 \sigma))]/2$. If $\sigma=\Theta(n)$, $\Lambda_\sigma$ places a constant fraction of its mass on low Hamming-weight frequencies, while if $\sigma=\mathcal{O}(1)$, low Hamming-weight frequencies receive exponentially small probability mass.

\subsection{Impact of IQP-QCBM's architecture on the presence of BPs} \label{subsec:uniform}
We show that, for each Fourier frequency $\vta$, the scaling of the variance of the characteristic function's partial derivative is completely determined by a simple architectural quantity of the IQP-QCBM: the ``critical rank'' $r^\vta$, which is the number of $\mathbb{F}_2$-linearly independent bitstrings in the generator set $S^\vta$.

\begin{theorem}[Critical rank determines the variance of the characteristic function's partial derivative]
\label{thm:uniform}
Consider an $n$-qubit IQP-QCBM in~\cref{def:IQP} with the generator set $S=\{\vtsj\}_{j=1}^D$, and let the reference measure on the parameter space be uniform, namely $\mu_n = \mathcal{U}([0,2\pi)^D)$. For a frequency $\vta \in \mathbb{F}_2^n$, define $S^{\vta} \coloneqq \{ \vtsj\in S \mid \vta \cdot \vtsj = 1 \}$
and the critical rank $r^\vta \coloneqq \dim \operatorname{span}_{\mathbb{F}_2}( S^{\vta})$
(so $r^\vta=0$ if $S^\vta=\emptyset$, in particular $r^{\mathbf 0}=0$). Then
\begin{equation}
\operatorname*{\mathbb{E}}_{\vttheta\sim\mu_n}\big[C_\vttheta(\vta)^2\big] = 2^{-r^\vta},
\end{equation}
and if $S^\vta\neq\emptyset$, then $\operatorname*{\mathbb{E}}_{\vttheta}[C_\vttheta(\vta)]=0$, hence
$\operatorname*{Var}_{\vttheta\sim\mu_n}[C_\vttheta(\vta)] = 2^{-r^\vta}$.
(If $S^\vta=\emptyset$, then $C_\vttheta(\vta)\equiv 1$.) For $\ell \in [D]$, if $\vtsl \notin S^\vta$, then $ \frac{\partial C_\vttheta(\vta)}{\partial \theta_\ell}=0$. If $\vtsl \in S^\vta$, then
\begin{equation}
\operatorname*{Var}_{\vttheta\sim\mu_n} \left[ \frac{\partial C_\vttheta(\vta)}{\partial \theta_\ell} \right] = 2^{\,2-r^\vta}. \end{equation}
\end{theorem}

The critical rank $r^\vta$ depends on the frequency $\vta$, as well as on the type and number of generators in the model. The variances above are monotone non-increasing under the addition of new generators: adding generators can only leave $r^\vta$ unchanged or increase it, and hence can only leave the variances unchanged or decrease them. This is consistent with the general expressivity–trainability trade-off observed in quantum supervised learning~\citet{mcclean18, Holmes22}. In contrast to more general quantum models, where gradient variances are often bounded only asymptotically as functions of circuit depth, IQP-QCBMs allow an exact characterization: once the generator set is specified, the critical ranks and therefore the variances can be computed exactly.

To develop intuition, we now examine several representative IQP-QCBM architectures. For each architecture, we compute or bound the critical rank, and then determine whether the model is trainable under certain kernel choices. We start with a simple architecture.

\begin{example}[Product state]
\label{eg:product}
Consider the generator set containing all the single-qubit generators and no other generators. Fix a frequency $\vta \in \mathbb{F}_2^n \setminus \{\bm{0}\}$. The generators that anticommute with $X^\vta$ are exactly those whose support overlaps with the support of $X^\vta$. Thus, $S^\vta=\{\bm{e}_k \mid a_k=1\}$, where $\bm{e}_k$ denotes the vector with $1$ in the $k$-th position and zeros elsewhere. Since these basis vectors $\bm{e}_k$ are $\mathbb{F}_2$-linearly independent, we have $r^\vta = \abs{\vta}$. Hence, low Hamming-weight frequencies with $\abs{\vta} \in \mathcal{O}(\log(n))$ have polynomially decaying variances, while high Hamming-weight frequencies with $\abs{\vta} \in \Theta(n)$ have exponentially decaying variances.
\end{example}

The product-state architecture produces no entanglement and can only express product distributions. Nevertheless, even in this simple model, characteristic function values at high Hamming-weight frequencies already decay exponentially. This is not unique to IQP-QCBMs; rather, it is an instance of BPs induced by global observables, a phenomenon studied previously in other quantum machine learning settings~\citet{Cerezo21,rudolph24}. Next, we consider a more expressive generator set.

\begin{example}[2D lattice]
\label{eg:lattice}
Let $G=(V,E)$ be a two-dimensional rectangular lattice graph on $n=\abs{V}$ qubits. Without loss of generality, we identify $V$ with $[n]$. Consider the generator set containing all single-qubit generators and all nearest-neighbor two-qubit generators: $\{\bm{e}_v : v\in V\} \cup \{\bm{e}_u+\bm{e}_v : (u,v)\in E\}$. Fix a Fourier frequency $\vta \in \mathbb{F}_2^n \setminus \{\bm{0}\}$, and let
\begin{equation*}
    A:=\{v\in V: a_v=1\}
\end{equation*}
be the support of the observable $X^\vta$. For a single-qubit generator $\bm{e}_v$, we have that $\bm{e}_v\in S^\vta$ iff $v\in A$. For a two-qubit generator $\bm{e}_u+\bm{e}_v$, we have that $\vta \cdot (\bm{e}_u+\bm{e}_v)=a_u+a_v\ (\mathrm{mod}\ 2)$. Hence $\bm{e}_u+\bm{e}_v\in S^\vta$ iff exactly one of $\{u,v\}$ lies in $A$. Equivalently, the two-qubit generators in $S^\vta$ are precisely the boundary edges of the cut $(A, V\setminus A)$.

Define the open neighborhood of $A$ by
\begin{equation}
B := \{u\in V\setminus A: \exists v\in A  \text{ s.t. } (u,v)\in E\},
\end{equation}
and the closed neighborhood by $N[A] = A \cup B$.
One can check that $r^\vta = \abs{N[A]}$, which depends not only on $\abs{\vta}$, but also on the topology of $A$. For the BP analysis, it is convenient to simply bound the quantity $\abs{N[A]}$.
From above, we can use the fact that each $v\in A$ has at most $4$ distinct neighbors. Hence, we have
\begin{equation}
\abs{\vta}=\abs{A} \leq  r^\vta=\abs{N[A]} \leq \abs{A}+4\abs{A} = 5\abs{\vta}.
\end{equation}
Hence, low Hamming-weight Fourier frequencies with $\abs{\vta} \in \mathcal{O}(\log(n))$ will have polynomially decaying variances, while high Hamming-weight frequencies with $\abs{\vta} \in \Theta(n)$ will have exponentially decaying variances.
\end{example}

The lattice architecture considered in \cref{eg:lattice} is already computationally interesting because the model becomes universal for quantum computation under postselection \cite{Fujii17}. Consequently, by standard postselection-based hardness arguments, the existence of an efficient zero-error classical weak simulator would imply a collapse of the polynomial hierarchy to its third level. As a side remark, if no single-qubit generators are included, then the output distributions would be efficiently classically simulable \cite{Fujii17}.

\begin{example}[Sparse Erd\H os-R\'enyi graph]
\label{eg:sparse}
Consider a generator set containing all single-qubit generators, and a random subset of two-qubit generators. Specifically include each two-qubit generator $\bm{e}_u+\bm{e}_v$ independently with probability $p=c \ln(n)/n$ for some scaling constant $c=\Theta(1)$. Equivalently, let $G=(V,E)$ be an Erd\H os-R\'enyi graph with edge probability $c \ln(n)/n$, and include $\bm{e}_u+\bm{e}_v$ iff $(u,v)\in E$. As in \cref{eg:lattice}, for a fixed Fourier frequency $\vta \in \mathbb{F}_2^n \setminus \{\bm{0}\}$, let $A$ be the support of the observable $X^\vta$, and denote the open and closed neighborhoods of $A$ by $B$ and $N[A]$ respectively. Again, we have $r^\vta = \abs{N[A]}$. For each vertex $u\notin A$, the event $\{u\notin B\}$ occurs iff $u$ has no edge to any vertex in $A$. Therefore, $\Pr[u\notin B] = (1-p)^{\abs{\vta}}$.
Also, the events $\{u\in B\}$ are independent for each $u \notin A$.
Therefore, $\abs{B} \sim \mathrm{Binomial}(n-\abs{\vta},\,q)$, 
where $q = 1-(1-p)^{\abs{\vta}}$ and \begin{equation}\operatorname*{\mathbb{E}}_G[r^\vta] = \abs{\vta} + (n-\abs{\vta})q. \label{eq:sparse-rank}
\end{equation}
Using that $(1-p)^\abs{\vta} = \exp(-p\abs{\vta} + \mathcal{O}(p^2\abs{\vta}))$, we can get that $q=p\abs{\vta} + \mathcal{O}((p\abs{\vta})^2)$ and $\operatorname*{\mathbb{E}}_G[r^\vta] \in \mathcal{O}(\log(n)^2)$ if $\abs{\vta} \in \mathcal{O}(\log(n))$, while $q = \Theta(1)$ and $\operatorname*{\mathbb{E}}_G[r^\vta] \in \Theta(n)$ if $\abs{\vta} \in \Theta(n)$. By Chernoff's bound, with high probability, for a single Erd\H os-R\'enyi graph, low Hamming-weight Fourier frequencies with $\abs{\vta} \in \Theta(1)$ will have polynomially decaying variances, ($\abs{\vta} \in \mathcal{O}(\log(n))$ gives subexponential but not polynomial), while high Hamming-weight frequencies with $\abs{\vta} \in \Theta(n)$ will have exponentially decaying variances.
\end{example}

The architecture considered in \cref{eg:sparse} is closely related to the ``sparse IQP circuits'' introduced by \citet{Bremner17}, except that here the parameters $\vttheta$ are drawn from a continuous uniform distribution, whereas their construction samples $\vttheta$ uniformly from a discrete set. In particular, they showed that the existence of an efficient classical sampler that can output distributions within small total variation distance to the true distributions of a constant fraction of instances from their ensemble of ``sparse IQP circuits'' would imply a collapse of the polynomial hierarchy to its third level~\citep{Bremner17}.

\begin{example}[Complete graph] \label{eg:complete} Consider the generator set containing all single-qubit generators and all two-qubit generators. Fix a frequency $\vta \in \mathbb{F}_2^n \setminus \{\bm{0}\}$, and let $A$ denote its support. Choose any $i\in A$. Then, for every $j\notin A$, the two-qubit generator $\bm{e}_i+\bm{e}_j$ belongs to $S^\vta$. These two-qubit generators, together with the single-qubit generators supported on $A$, form a basis. Therefore, for any frequency $\vta \in \mathbb{F}_2^n \setminus \{\bm{0}\}$, $r^\vta =n$, and the variances decay exponentially.
\end{example}

The lattice and sparse Erd\H os-R\'enyi architectures in \cref{eg:lattice} and \cref{eg:sparse} have larger generator sets than the product-state architecture, but they still retain polynomially decaying variances for low Hamming-weight frequencies. Therefore, by choosing a kernel with sufficiently large weights on the low Hamming-weight part of the spectrum, one can avoid BPs in the MMD for these architectures. However, for architectures with even larger generator sets, such as the complete graph in \cref{eg:complete}, every frequency $\vta \in \mathbb{F}_2^n \setminus \{\bm{0}\}$ has full critical rank and hence exponentially small variance. Then no choice of kernel can prevent the MMD gradient from exhibiting a BP, as implied by \cref{prop:upper}.

\subsection{Classical non-simulability and trainability}\label{subsec:non-simulability}
One of the central questions for quantum generative models is whether their output distributions can be efficiently simulated on a classical computer. For IQP circuits, known sampling-hardness proofs~\cite{Bremner17} require the circuit ensemble to satisfy an anti-concentration property.

\begin{definition}[$L^2$ anti-concentration of discrete distribution family]
\label{def:anti-concentration}
Let $Q = \{q_\vttheta\}_\vttheta$ be a family of probability distributions on $\{0,1\}^n$ and let $\vttheta \sim P$ for some distribution. We say $Q$ exhibits $L^2$ anti-concentration if there exists a constant $b>0$ such that
\begin{equation}
    \operatorname*{\mathbb{E}}_{\vttheta \sim P}\left[\sum_{x\in\{0,1\}^n} q_\vttheta(x)^2\right] \leq b/2^n,
\end{equation}
where the quantity $\sum_x q_\vttheta(x)^2$ is known as the collision probability.  
\end{definition}
\cref{def:anti-concentration} implies an alternative formulation of anti-concentration: there exist $\alpha>0,\beta>0$ such that $\forall x\in \{0,1\}^n, \;\Pr\left[q_\vttheta(x)\geq \alpha 2^{-n}\right]\ge \beta$,
which appears in sampling hardness arguments of IQP circuits~\cite{Bremner16, Bremner17}. To see this implication, we first observe that due to symmetry, all bitstrings $x$ have equal contributions: $\E_\vttheta[\sum_xq_\vttheta(x)] = 1\implies \E_\vttheta[q_\vttheta(x)] = 2^{-n}$ and $\E_\vttheta[q_\vttheta(x)^2] \le b2^{-2n}$, then we apply the Paley-Zygmund inequality to obtain $\Pr[q_\vttheta(x) \ge \alpha2^{-n}] \ge (1-\alpha)^2/b$. Setting $\beta = (1-\alpha)^2/b$ completes the proof.

To study the classical hardness of IQP-QCBMs, it is convenient to relate the $L^2$ anti-concentration condition in~\cref{def:anti-concentration} to the critical rank quantities introduced previously  in~\cref{thm:uniform}. Since the Fourier transform is an $L^2$ isometry (the Parseval identity), we can achieve it automatically in the next proposition.

\begin{proposition}[$L^2$ anti-concentration of IQP-QCBM]
\label{lem:parseval}
Let $Q=\{q_\vttheta\}_\vttheta$ be the ensemble of output distributions of an
$n$-qubit IQP-QCBM with $\vttheta\sim\mathcal U([0,2\pi)^D)$, and let
$\mathcal Z \coloneqq \sum_{\vta\in\mathbb{F}_2^n} 2^{-r^\vta}$, with $r^{\mathbf 0}=0$.
Then the expected collision probability is exactly
\begin{equation}\label{eq:collision-identity}
\operatorname*{\mathbb{E}}_{\vttheta}\Big[\sum_{x\in\{0,1\}^n} q_\vttheta(x)^2\Big]
= \frac{\mathcal Z}{2^n}.
\end{equation}
Consequently, $Q$ is $L^2$ anti-concentrated with constant $b$ if and only if
$\mathcal Z\le b$. 
\end{proposition}

Using this proposition, we can check whether the examples of IQP architectures studied in~\cref{subsec:uniform} have the anti-concentration property or not. 
\begin{itemize}
    \item For~\cref{eg:product}, we have: $\sum_{\vta \in \mathbb{F}_2^n}  2^{-r^\vta}=\sum_{k=0}^n \binom{n}{k}2^{-k}=(3/2)^n$, so the output distribution is not anti-concentrated.
    \item For~\cref{eg:lattice}, we have: $\sum_{\vta \in \mathbb{F}_2^n} 2^{-r^\vta} \geq \sum_{k=0}^n \binom{n}{k}2^{-5k}=\left(\frac{33}{32}\right)^n$, so the output distribution is not anti-concentrated.
    \item For~\cref{eg:sparse}, the critical rank $r^{\vta}(G)$ varies for each graph instance $G$. The distribution of graph instances depends on the scaling constant $c$. In \cref{thm:erm-sharp} we derive an exact formula for $\mathbb{E}_G[\mathcal Z(G)]$ for any scaling constant $c>0$, where we write $\mathcal Z(G)=\sum_{\vta\in\mathbb F_2^n}2^{-r^\vta(G)}$. We show that $\mathbb{E}_G[\mathcal Z(G)]$ stays bounded when $c\geq2$ and diverges when $c<2$. Then, in \cref{cor:erm-typical} we extend the results to $\mathcal Z(G)$ of fixed graphs. As a side remark, the sharp anti-concentration threshold at $p=2\ln(n)/n$ is larger than the more well-known connectivity threshold at $p=\ln(n)/n$. Intuitively, the reason is that while connectivity only requires each qubit to have at least one neighbor, anti-concentration requires each qubit to have a larger number of neighbors.
\item For~\cref{eg:complete}, we have: $\sum_{\vta \in \mathbb{F}_2^n} 2^{-r^\vta}= 1+(2^n-1)2^{-n} = 2 - 2^{-n}$, so the output distribution is anti-concentrated.
\end{itemize}

\begin{theorem}\label{thm:erm-sharp}
Let $c>0$ and $p=c\ln n/n$. Let $G\sim G(n,p)$ be an Erd\H{o}s--R\'enyi graph instance. We write $\mathcal Z(G)=\sum_{\vta\in\mathbb F_2^n}2^{-r^\vta(G)}$. The expected quantity $\mathbb{E}_G[\mathcal Z(G)]$ can be bounded by:
\begin{equation}\label{eq:erm-bounds}
2\sum_{k=0}^{K}\frac{\big(T_n(1-o(1))\big)^k}{k!}-\mathcal O\big(n^{2K}2^{-n}\big)\ \le\ \operatorname*{\mathbb{E}}_G\big[\mathcal Z(G)\big]\ \le\ 2\exp\big(T_n(1+o(1))\big)+e^{-\Omega(n/\ln^3n)}
\end{equation}
for every fixed $K\in\mathbb N$, where $T_n=n2^{-n}(2-p)^{n-1}$. In particular, a threshold occurs at $c=2$:
\begin{equation}
    \begin{cases}
    c<2 & \lim\limits_{n\to \infty} T_n = \infty, \lim\limits_{n\to \infty}\mathbb{E}_G[\mathcal Z(G)] = \infty \\
    c=2 & \lim\limits_{n\to \infty}T_n = \frac{1}{2}, \lim\limits_{n\to \infty}\mathbb{E}_G[\mathcal Z(G)] = 2\sqrt{e}\\
    c>2 & \lim\limits_{n\to \infty}T_n = 0, \lim\limits_{n\to \infty}\mathbb{E}_G[\mathcal Z(G)] = 2 
\end{cases}
\end{equation}
\end{theorem}

\begin{corollary}[Typical graphs]\label{cor:erm-typical}
Let $G\sim G(n,p)$ with $p=c\ln n/n$.

\noindent\textup{(a)} If $c>2$, then with high probability $2-2^{-n}\le\mathcal Z(G)\le2+n^{-(c-2)/4}$.

\noindent\textup{(b)} If $c=2$, then for every fixed $\delta\in(0,1)$, $\Pr_G\big[\mathcal Z(G)\le2+(2\sqrt e-2)/\delta\big]\ge1-\delta-o(1)$.

\noindent\textup{(c)} If $c<2$, then with high probability $\mathcal Z(G)\ge(1-o(1))\,T_n\to\infty$.
\end{corollary}

Example~\ref{eg:sparse} and~\cref{thm:uniform} show that at low frequencies $\abs{\vta}\in\mathcal{O}(1)$, the sparse Erd\H{o}s--R\'enyi circuits have logarithmic critical rank and hence polynomially decaying gradient variance with high probability. We therefore ask whether this low-rank fraction of the sparse Erd\H{o}s--R\'enyi circuits is also hard to sample classically. We first show the tail bound of the critical rank.

\begin{lemma}[Critical rank for sparse IQP circuits at low frequency] \label{lemma:low-rank-tail}
Fix $k_0 \geq 1$ independent of $n$ and a sufficiently large constant
$c>2$. For $G\sim G(n,c\ln n/n)$, we define $R_{k_0}(G) \coloneqq \max_{1\leq\abs{\vta}\leq k_0}r^{\vta}(G)$.
There exist constants $K,\gamma>0$ such that $E_n \coloneqq \{G \colon R_{k_0}(G)\leq K\ln n\}$ satisfies $\delta_n \coloneqq\Pr_{G\sim G(n,c\ln n/n)}[G\notin E_n] \in \mathcal{O}(n^{-\gamma})$.
\end{lemma}
The proof is in \cref{app:low-rank-tail}.

We denote by $|E(G)|$ the number of edges in $G$ and by $d_G = n + |E(G)|$ the number of parameters in the circuit. For a fixed connectivity graph $G$, the corresponding diagonal unitary of IQP is $U(G,\vttheta)=\prod_{i=1}^n e^{i\theta_i Z_i}\prod_{\{i,j\}\in E(G)}e^{i\theta_{ij}Z_iZ_j}$. Next, we consider two probability distributions of circuits and parameters in the sampling-hardness result. 

Both distributions sample the circuit's connectivity graph $G$ from the sparse Erd\H os-R\'enyi family, i.e., $G\sim G(n,p)$, where $p=c\ln n/n$ with $c>2$, and differ in the way $\theta$ is sampled: 
\begin{itemize}
    \item $\mathcal{D}_n$ (adopted from~\cite{Bremner17}) samples $\theta_{ij}=\pi k / 8, k\sim\mathcal{U}(\{0, 1,2,3\})$ for a two-qubit gate supported on edge $(i,j)$, and samples $\theta_{i}=-\pi k / 8 - \sum_{j\colon (i,j)\in E(G)} \theta_{ij}, k\sim\mathcal{U}(\{0, 1,\ldots,7\})$ for a single-qubit gate on vertex $i$. See~\cref{app:sampling-hardness} for its equivalence to the original definition in~\cite{Bremner17}.
    \item $\mathcal{C}_n$ samples $\theta$ u.a.r.~in $[0,2\pi)^{d_G}$, i.e., $\theta \mid G\sim \mathcal{U}([0,2\pi)^{d_G})$.
\end{itemize}
A circuit instance/sample is a pair $C=(G,\vttheta)$. We write the output distribution of $C$ as $P_C$. Restricting the above distributions to the event $E_n$ induces the
conditional measures $\mathcal{D}_n^{\rm low} \coloneqq \mathcal{D}_n(\cdot\mid G\in E_n)$ and $\mathcal{C}_n^{\rm low} \coloneqq \mathcal{C}_n(\cdot\mid G\in E_n)$. By Proposition~\ref{lem:parseval} and \cref{thm:erm-sharp}, we have $\mathbb{E}_{C\sim\mathcal{C}_n^{\rm low}}\sum_x P_C(x)^2 \leq 2^{-n}\mathbb{E}_G[\mathcal Z(G)]/(1-\delta_n) =(2+o(1))2^{-n}$, so the low-rank ensemble $\mathcal{C}_n^{\rm low}$ is anti-concentrated.

\begin{theorem}[Sampling hardness of low-rank sparse IQP circuits]
\label{thm:sampling-hardness}
Assume Conjecture~8 of \cite{Bremner17} and noncollapse of the polynomial
hierarchy. Fix $k_0\in\mathcal{O}(1)$ and choose $K$ as above. There exist constants $\varepsilon>0$ and $\eta>0$ such that
\begin{enumerate}
    \item There is no uniform polynomial-time classical sampler $A$ which takes a description of circuit $C$ and outputs a distribution $Q_C^A$ satisfying
    \begin{equation} \label{eq:low-rank-average-hardness}
    \Pr_{C\sim\mathcal D_n^{\rm low}}\left(\|Q_C^A-P_C\|_1\leq\varepsilon \right) \geq 1 - \eta
    \end{equation}
    for all sufficiently large $n$.
    \item There is no uniform polynomial-time classical sampler $A$ whose output distribution $Q_C^A$ satisfies
    \begin{equation}
    \|Q_C^A-P_C\|_1\leq\varepsilon \quad\text{for every } C\in\mathcal{C}_n^{\rm low}
    \label{eq:low-rank-sampling-error}
    \end{equation}
    for all sufficiently large $n$.
\end{enumerate}
\end{theorem}

We remark that with \cref{eg:sparse,lemma:low-rank-tail,thm:sampling-hardness}, we identify that for $p=c\ln n/n$ with sufficiently large $c>2$, a sparse Erd\H{o}s--R\'enyi IQP-QCBM belongs with high probability to the low-critical-rank family and, with a kernel supported on constant-weight frequencies, is trainable in the sense of avoiding a BP. At the same time, under the sparse average-case hardness conjecture and noncollapse of the polynomial hierarchy, this conditioned family is classically non-simulable in the worst case.

\section{Numerical validation} \label{sec:numerical-validation}
We provide two experimental validations of our theoretical predictions (see~\cref{sec:setup} for the experimental setup). First, we verify the scaling of the MMD partial derivative with respect to the number of qubits on MNIST-inspired synthetic datasets of size $n=20,40,\ldots,200$, whose marginals and correlations approximately match the Fourier spectrum of binarized MNIST. We test the five IQP architectures considered in the paper: product state, 2D lattice, sparse ER graph with $c=0.5$, sparse ER graph with $c=2$, and complete graph. We consider three kernels supported on frequencies of Hamming weight $2$, $\log(n)$, and $n/4$, respectively. The results are shown in~\cref{pic:numerics} (top row). Due to finite-sample estimation of the MMD gradient, the experiments exhibit a numerical noise floor, explaining why several curves flatten below $10^{-9}$. Across all kernels and architectures, the observed gradient scaling is either polynomial or exponential in $n$, in agreement with the predictions in \cref{eg:product,eg:lattice,eg:sparse,eg:complete}.

Second, we study MMD loss trajectories to test whether the predicted presence or absence of BPs matches practical trainability. We show the MMD training trajectories in~\cref{pic:numerics} (bottom row). The 2D lattice and sparse ER graph architectures exhibit improved training stability and convergence as the sample size for gradient estimation increases. Moreover, the 2D lattice architecture consistently converges faster than the sparse ER graph architectures, matching its smaller critical rank predicted in \cref{eg:lattice,eg:sparse}. In contrast, the complete-graph architecture remains untrainable even at the largest sample size, consistent with the predicted BP.

Note that the sparse ER graph with $c=0.5$ lies below the threshold and is therefore not anti-concentrated, while the one with $c=2$ is at the anti-concentration threshold. However, the trainability prediction for low-weight frequencies holds for any scaling constant $c$ (\cref{eg:sparse}).

\begin{figure}[h]
\centering
\begin{subfigure}{\textwidth}
\includegraphics[width=\textwidth]
{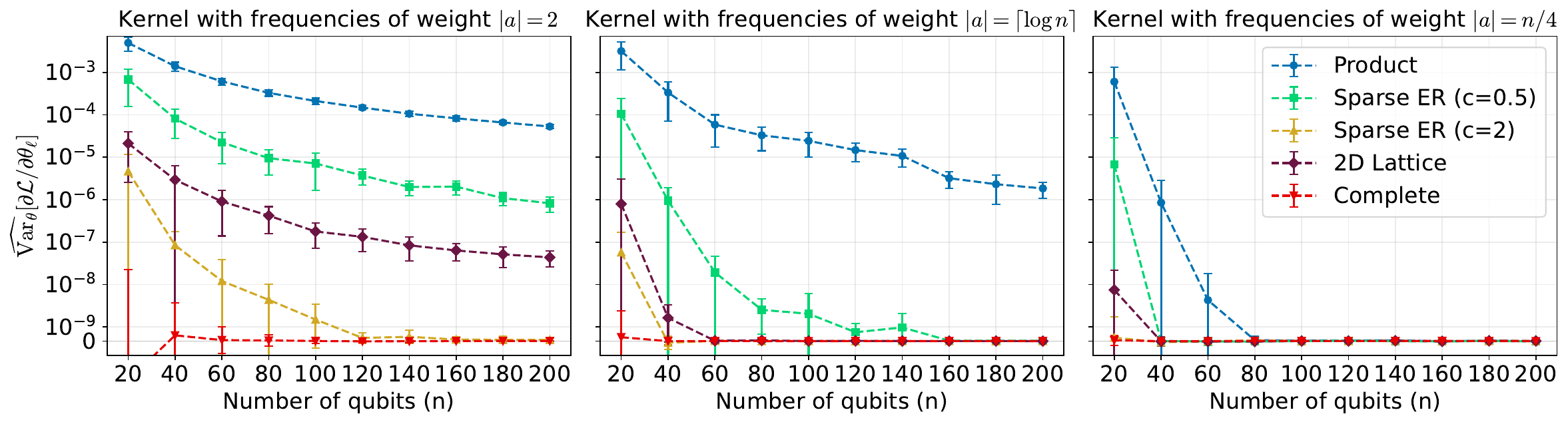}
\end{subfigure}
\begin{subfigure}{\textwidth}
\includegraphics[width=\textwidth]
{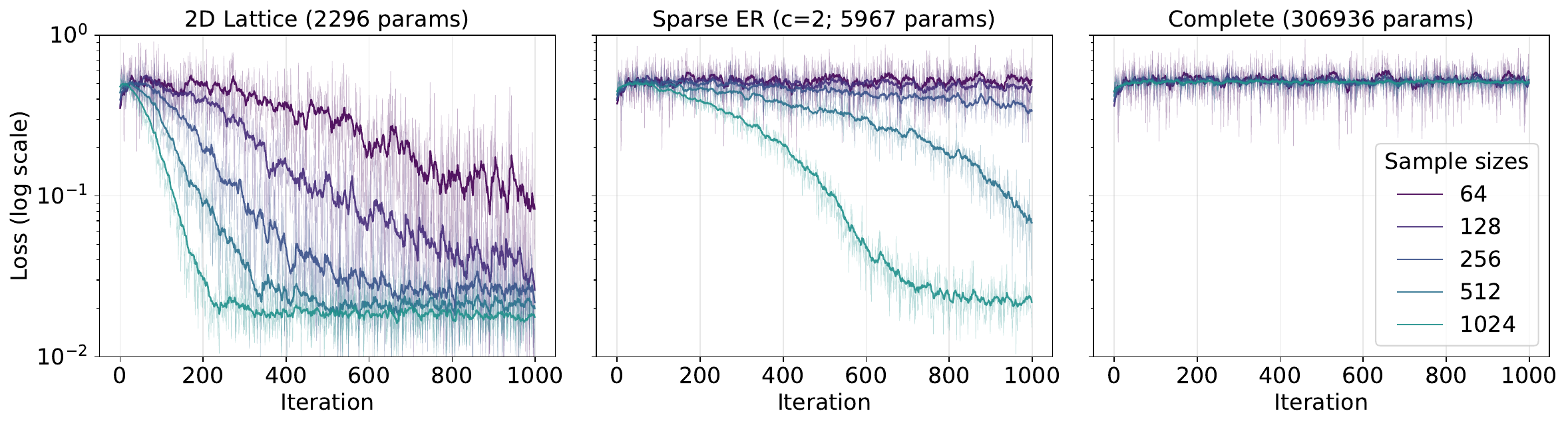}
\end{subfigure}
\caption{Top: Scaling of the MMD derivative with respect to the number of qubits, measured on MNIST-inspired datasets and reporting the mean and standard deviation over $50$ repetitions. Bottom: Trajectories of the MMD training loss for learning the $784$-qubit binarized MNIST dataset, with increasing sample sizes for estimating the MMD loss.}
\label{pic:numerics}
\end{figure}

\section{Conclusions and Limitations}
\label{sec:conclusion}
In this work, we studied the trainability and classical non-simulability of IQP-QCBMs. We expressed the MMD-gradient variance in terms of the IQP output's characteristic function and the kernel's spectral measure. Our first major result identifies the critical rank, which exactly determines the scaling of characteristic-function gradient variances and thereby characterizes the presence or absence of BPs. Unlike general depth-based bounds, critical ranks can be computed exactly from the generator set. We demonstrated the framework for product-state, two-dimensional lattice, sparse Erdős--Rényi, and complete-graph architectures.

Our second major result concerns sparse Erd\H os-R\'enyi architectures at a suitable edge probability. Their low-weight gradients decay only polynomially, while the conditioned family is classically non-simulable in the worst case, assuming sparse average-case hardness and noncollapse of the polynomial hierarchy. Choosing a kernel concentrated on these frequencies therefore allows trainability and classical non-simulability to coexist. This contrasts with the trade-off suggested for supervised QML models, where trainability often implies classical simulability.

This does not contradict the claim of \citet{herbst25} that sampling hardness rooted in anti-concentration precludes trainability. Their argument assumes pseudo-independence, where $p(x)=Y_x/(\sum_jY_j)$ for i.i.d. positive $Y_j$. Pseudo-independence implies anti-concentration, but not conversely, and forces all nonzero characteristic functions to vanish exponentially. By contrast, $L^2$ anti-concentration controls only the collision probability and permits a small set of frequencies to remain polynomially visible, as illustrated by \cref{eg:sparse}.

Here we discuss the limitations of our work. First, our trainability analysis focuses on gradient scaling under random parameter initialization. Second, avoiding a barren plateau is a necessary, but not sufficient, condition for efficient end-to-end optimization. Third, the MMD lower bound is averaged over a target-distribution ensemble, and its implications for a particular dataset depend on the dataset's Fourier spectrum and the chosen kernel. Lastly, the classical sampling-hardness result is conditional on standard complexity-theoretic conjectures. Extending the analysis to noisy devices, broader initialization schemes, and application-specific targets are interesting future directions.

\paragraph{Acknowledgements} VD and HW were supported by the Dutch National Growth Fund (NGF), as part of the Quantum Delta NL programme.
This work was also supported by the European Union’s Horizon Europe program through the ERC CoG BeMAIQuantum (Grant No. 101124342). This research was funded by the BMW Group. 

\newpage
\bibliographystyle{plainnat}
\bibliography{ref}

\appendix

\onecolumn
\section{Proof of Theorems}
\subsection{Proof of~\cref{prop:lower}}
Consider an $n$-qubit IQP-QCBM in~\cref{def:IQP}, with generator set $S=\{\vtsj\}_{j=1}^D$, and the MMD loss $\mathcal{L}$ defined in~\cref{prop:mmd-loss}. Let $\mu_n$ be a reference measure on the parameter space that is symmetric about 0 (e.g., uniform), and let $\Lambda$ be a positive spectral probability measure on $\mathbb{F}_2^n$. Then the variance of the partial derivative of the MMD admits the decomposition
\begin{align*}
&\operatorname*{Var}_{\vttheta\sim \mu_n} \left[ \frac{\partial \mathcal{L}(\vttheta)}{\partial \theta_\ell} \right]
= 4 \operatorname*{\mathbb{E}}_{\vttheta\sim \mu_n} \left( \sum_{\vta, \vtb \in \mathbb{F}_2^n} \Lambda(\vta) \Lambda(\vtb)F(\vta, \vtb)\right)\\
&F(\vta, \vtb) = \left( C_\vttheta(\vta) C_\vttheta(\vtb) - 2C_p(\vta) C_\vttheta(\vtb) + C_p(\vta) C_p(\vtb) \right)\frac{\partial C_\vttheta(\vta)}{\partial \theta_\ell} \frac{\partial C_\vttheta(\vtb)}{\partial \theta_\ell} \nonumber
\end{align*}
\begin{proof}
We perform expansion of the term directly, using \cref{prop:mmd-loss}.
\begin{align*}
    \operatorname*{Var}_{\vttheta} \left[ \frac{\partial \mathcal{L}(\vttheta)}{\partial \theta_\ell} \right] &=    \operatorname*{Var}_{\vttheta} \left[ \frac{\partial }{\partial \theta_\ell}  \left(\sum_{\vta \in \mathbb{F}_2^n} \Lambda(\vta)\left(C_p(\vta) - C_\vttheta(\vta) \right)^2\right)\right]
\\ &= \operatorname*{Var}_{\vttheta} \left[-2 \sum_\vta \Lambda(\vta) \left(C_p(\vta) - C_\vttheta(\vta) \right) \frac{\partial C_\vttheta(\vta)}{\partial \theta_\ell} \right] \\
&= 4 \sum_{\vta, \vtb} \Lambda(\vta) \Lambda(\vtb) \operatorname*{Cov}_{\vttheta}\left(\left(C_p(\vta) - C_\vttheta(\vta)\right) \frac{\partial C_\vttheta(\vta)}{\partial \theta_\ell} , \left(C_p(\vtb) - C^\vtb_\vttheta\right) \frac{\partial C^\vtb_\vttheta}{\partial \theta_\ell} \right) \\
&= 4 \sum_{\vta, \vtb} \Lambda(\vta) \Lambda(\vtb) \bigg(\operatorname*{\mathbb{E}}_{\vttheta} \left[ (C_p(\vta) - C_\vttheta(\vta)) \frac{\partial C_\vttheta(\vta)}{\partial \theta_\ell} (C_p(\vtb) - C^\vtb_\vttheta) \frac{\partial C^\vtb_\vttheta}{\partial \theta_\ell} \right] \\
&- \operatorname*{\mathbb{E}}_{\vttheta}\left[(C_p(\vta) - C_\vttheta(\vta)) \frac{\partial C_\vttheta(\vta)}{\partial \theta_\ell} \right] \operatorname*{\mathbb{E}}_{\vttheta}\left[(C_p(\vtb) - C^\vtb_\vttheta) \frac{\partial C^\vtb_\vttheta}{\partial \theta_\ell} \right] \bigg) \\
&= 4 \operatorname*{\mathbb{E}}_{\vttheta} \left[\sum_{\vta, \vtb} \Lambda(\vta) \Lambda(\vtb)  \left( C_p(\vta) - C_\vttheta(\vta) \right) \frac{\partial C_\vttheta(\vta)}{\partial \theta_\ell} \left(C_p(\vtb) - C^\vtb_\vttheta\right) \frac{\partial C^\vtb_\vttheta}{\partial \theta_\ell} \right] \\
&= 4 \operatorname*{\mathbb{E}}_{\vttheta}\bigg[ \sum_{\vta, \vtb} \Lambda(\vta) \Lambda(\vtb) \bigg( C_p(\vta) C_p(\vtb) \frac{\partial C_\vttheta(\vta)}{\partial \theta_\ell} \frac{\partial C_\vttheta(\vtb)}{\partial \theta_\ell} - C_p(\vta) C_\vttheta(\vtb) \frac{\partial C_\vttheta(\vta)}{\partial \theta_\ell} \frac{\partial C_\vttheta(\vtb)}{\partial \theta_\ell}\\
& - C_p(\vtb) C_\vttheta(\vta) \frac{\partial C_\vttheta(\vta)}{\partial \theta_\ell} \frac{\partial C_\vttheta(\vtb)}{\partial \theta_\ell} + C_\vttheta(\vta) C_\vttheta(\vtb) \frac{\partial C_\vttheta(\vta)}{\partial \theta_\ell} \frac{\partial C_\vttheta(\vtb)}{\partial \theta_\ell}\bigg) \bigg]
\end{align*}
In the above derivation, many terms vanish. In particular, all terms of the forms $\frac{\partial C_\vttheta(\vta)}{\partial \theta_\ell}$ and $C_\vttheta(\vta)\frac{\partial C_\vttheta(\vta)}{\partial \theta_\ell}$ have zero expectation values. More specifically, these terms either vanish when $\vta \cdot\vtsl =0$ or are odd functions of $\theta_\ell$ when $\vta \cdot\vtsl =1$; their expectations vanish due to the symmetry of $\mu_n$.
\end{proof}
Consider averaging this quantity over an unbiased ensemble of target distributions $\mathcal P$. Specifically, we assume that, $\forall\vta, \vtb \in \mathbb{F}_2^n \setminus \{\bm{0}\},\operatorname*{\mathbb {E}}_{p\sim\mathcal {P}}[C_p(\vta)]=0$, and $\operatorname*{\mathbb {E}}_{p\sim\mathcal {P}}\!\left[C_p(\vta)C_p(\vtb)\right]=\delta_{\vta,\vtb}\,\sigma^2(\vta)$, where $\sigma^2(\vta) = \operatorname*{\mathbb {E}}_{p \sim \mathcal{P}}\left[C_p(\vta)^2\right]$. Under this ensemble average, we obtain the lower bound:
\begin{equation*}
\operatorname*{\mathbb{E}}_{p\sim\mathcal{P}} \left[\operatorname*{Var}_{\vttheta\sim\mu_n}  
\left[\frac{\partial \mathcal{L}(\vttheta)}{\partial \theta_\ell} \right]\right] \geq
4 \sum_{\vta\in \mathbb{F}_2^n} \Lambda^2(\vta)\sigma^2(\vta)\operatorname*{Var}_{\vttheta\sim\mu_n} \left[ \frac{\partial C_\vttheta(\vta)}{\partial \theta_\ell} \right].
\end{equation*}
\begin{proof}
Continuing from the proof above, we first identify the four-copy terms as a squared term, i.e.,
\begin{equation*}
    \operatorname*{\mathbb{E}}_{\vttheta} \left[ \sum_{\vta, \vtb} \Lambda(\vta) \Lambda(\vtb) C_\vttheta(\vta) C_\vttheta(\vtb) \frac{\partial C_\vttheta(\vta)}{\partial \theta_\ell} \frac{\partial C_\vttheta(\vtb)}{\partial \theta_\ell}\right] =\operatorname*{\mathbb{E}}_{\vttheta}\left[ \left(\sum_{\vta \in \mathbb{F}_2^n} \Lambda(\vta) C_\vttheta(\vta) \frac{\partial C_\vttheta(\vta)}{\partial \theta_\ell}\right)^2\right]
\end{equation*}
which we can safely drop in the lower bound as non-negative terms. 

Then we can apply the ensemble average. All the three-copy terms $C_p(\vta) C_\vttheta(\vtb) \frac{\partial C_\vttheta(\vta)}{\partial \theta_\ell} \frac{\partial C_\vttheta(\vtb)}{\partial \theta_\ell}$ vanish. Due to the uncorrelatedness between $C_p(\vta)$ and $C_p(\vtb)$, all the off-diagonal two-copy terms $C_p(\vta) C_p(\vtb) \frac{\partial C_\vttheta(\vta)}{\partial \theta_\ell} \frac{\partial C_\vttheta(\vtb)}{\partial \theta_\ell}$ also vanish. Finally, we obtain the desired non-trivial lower bound for an average-case problem.
\begin{align*}
 \operatorname*{\mathbb{E}}_p\left[\operatorname*{Var}_{\vttheta} \left[ \frac{\partial \mathcal{L}(\vttheta)}{\partial \theta_\ell} \right]\right] &\geq 4 \operatorname*{\mathbb{E}}_{\vttheta}  \left[ \sum_{\vta \in \mathbb{F}_2^n} \Lambda(\vta)^2 (C_p(\vta))^2 \left(\frac{\partial C_\vttheta(\vta)}{\partial \theta_\ell}\right)^2 \right] \\
 &= 4 \sum_{\vta\in\mathbb{F}_2^n}\Lambda(\vta)^2\sigma_\vta^2\operatorname*{Var}_{\vttheta} \left[ \frac{\partial C_\vttheta(\vta)}{\partial \theta_\ell} \right]
\end{align*}
\end{proof}

\subsection{Proof of \cref{prop:upper}}
Under the same conditions as in~\cref{prop:lower}, the variance of the partial derivative of the MMD satisfies
\begin{equation*}
    \operatorname*{Var}_{\vttheta \sim\mu_n} \left[ \frac{\partial \mathcal{L}(\vttheta)}{\partial \theta_\ell} \right] \leq 16 \sum_{\vta\in \mathbb{F}_2^n} \Lambda(\vta)\left[  \operatorname*{Var}_{\vttheta\sim\mu_n} \left[ \frac{\partial C_\vttheta(\vta)}{\partial \theta_\ell} \right]\right]
\end{equation*}
\begin{proof}
\begin{align*}
\operatorname*{Var}_{\vttheta} \left[ \frac{\partial \mathcal{L}(\vttheta)}{\partial \theta_\ell} \right] &\leq \operatorname*{\mathbb{E}}_{\vttheta} \left[ \left(\frac{\partial \mathcal{L}(\vttheta)}{\partial \theta_\ell} \right)^2 \right] \\
&= 4\operatorname*{\mathbb{E}}_{\vttheta} \left[  \operatorname*{\mathbb{E}}_{\vta}\left[\left( C_p(\vta) - C_\vttheta(\vta) \right) \frac{\partial C_\vttheta(\vta)}{\partial \theta_\ell} \right]^2  \right]\\
&\leq 4\operatorname*{\mathbb{E}}_{\vttheta} \left[  \operatorname*{\mathbb{E}}_{\vta}\left[\left( C_p(\vta) - C_\vttheta(\vta) \right)^2 \right] \operatorname*{\mathbb{E}}_{\vta}\left[\left(\frac{\partial C_\vttheta(\vta)}{\partial \theta_\ell} \right)^2\right]\right]\\
&\leq 4\operatorname*{\mathbb{E}}_{\vttheta} \left[  4 \operatorname*{\mathbb{E}}_{\vta}\left[\left(\frac{\partial C_\vttheta(\vta)}{\partial \theta_\ell} \right)^2\right]\right]\\
&= 16 \sum_\vta \Lambda(\vta)\left[  \operatorname*{Var}_{\vttheta} \left[ \frac{\partial C_\vttheta(\vta)}{\partial \theta_\ell} \right]\right]
\end{align*}
In the derivation above, we first used the Cauchy--Schwarz inequality and then the bound that, for a fixed $\vta$, $(C_p(\vta) - C_\vttheta(\vta))^2 \leq 4$. Lastly, we exchanged the order of expectations over $\vta$ and $\vttheta$ and used the fact that $\operatorname*{\mathbb{E}}_{\vttheta}\left[\partial C_\vttheta(\vta) / \partial \theta_\ell\right] = 0$.
\end{proof}

\subsection{Proof of \cref{thm:uniform}}
Before proving \cref{thm:uniform}, we will first prove a more general result applying to a broader family of reference measures that are symmetric about $0$, from which \cref{thm:uniform} follows as a special case.

We wish to show that: Let $\theta_1,\dots,\theta_D$ be i.i.d.\ with a common distribution that is symmetric about $0$. For a fixed Fourier frequency $\vta \in \mathbb{F}_2^n$, define $S^{\vta} := \{ \vtsj \big| \vta \cdot \vtsj = 1 \}$, $
m^\vta := \abs{S^{\vta}}$. Let $
\Xi^\vta:= \{J\subseteq S^\vta:\ \sum_{j\in J}\vtsj=\mathbf 0 \}.$ Then, 
\begin{equation*}
\operatorname*{\mathbb{E}}_\vttheta\big[C_\vttheta(\vta)\big] = \kappa^{m^\vta}, \qquad
\operatorname*{\mathbb{E}}_\vttheta\big[C_\vttheta(\vta)^2\big] = \sum_{J\in \Xi^\vta} \mu^{|J|} \nu^{m^\vta-|J|},
\end{equation*}
and hence
\begin{equation*}
     \operatorname*{Var}_\vttheta \left[ C_\vttheta(\vta) \right] = 
\left(\sum_{J\in \Xi^\vta} \mu^{|J|} \nu^{m^\vta-|J|}\right) - \kappa^{2m^\vta},
\end{equation*}
where $\mu=\operatorname*{\mathbb{E}}[\sin^2(2\theta_1)]$, $\nu=\operatorname*{\mathbb{E}}[\cos^2(2\theta_1)]$, $\kappa=\operatorname*{\mathbb{E}}[\cos(2\theta_1)]$. Fix $\ell\in[D]$. If $\vtsl \notin S^\vta$, then $ \frac{\partial C_\vttheta(\vta)}{\partial \theta_\ell}=0$. If $\vtsl \in S^\vta$, then, 
\begin{align*}
\operatorname*{Var}_\vttheta \left[ \frac{\partial C_\vttheta(\vta)}{\partial \theta_\ell} \right]
&= 4\sum_{J\in \Xi^\vta}
\Bigl[
\bm{1}\{\ell\in J\}\,\mu^{|J|-1}\nu^{m^\vta-|J|+1} \notag
\\&+
\bm{1}\{\ell\notin J\}\,\mu^{|J|+1}\,\nu^{m^\vta-|J|-1}\Bigr].
\label{eq:symmetric-variance}
\end{align*}

\begin{proof}
Fix $\vta\in\mathbb{F}_2^n$ and abbreviate $S:=S^\vta$, $m:=m^\vta$. If $m=0$, then $X^\vta$ commutes with $U_\vttheta$, so $C_\vttheta(\vta)\equiv 1$; this agrees with all formulas above, since then $\Xi^\vta=\{\emptyset\}$ and $\kappa^{0}=1$. Assume $m\geq 1$ from now on. Start from the expression $C_\vttheta(\vta)=\bra{+} U_\vttheta^\dagger X^{\vta}U_\vttheta\ket{+}
$. By recursively swapping the gates in $U_\vttheta$ with the observable $X^{\vta}$ and applying the commutation or anticommutation rule, one can derive the following:
\begin{align*}
C_\vttheta(\vta)
&= \bra{+}^{\otimes n} \prod_{\vtsj\in S^\vta} \exp(-2i\theta_jZ^\vtsj) \ket{+}^{\otimes n} \\
&= \bra{+}^{\otimes n} \prod_{\vtsj\in S^\vta} (\cos(2\theta_j)I - i \sin(2\theta_j)Z^\vtsj) \ket{+}^{\otimes n}  \\
&= \sum_{J\in\Xi^\vta} R_J(\vttheta), \qquad R_J(\vttheta)=\eta_J 
\Bigl(\prod_{j\in J} \sin(2\theta_j)\Bigr)\Bigl(\prod_{j\in S\setminus J}\cos(2\theta_j)\Bigr),
\end{align*}
where $\eta_J \in \{-1,1\}$. The equation is a reformulation of equation 27 of \cite{recioarmengol25}. Note that $\forall J \in \Xi^\vta, \abs{J} = 0 \, (\text{mod 2})$, so no imaginary number $i$ appears. Also note that $\sin(\cdot)$ is an odd function, whose expectation vanishes for any symmetric distribution. We will use this property many times.

First, we derive the variance of the characteristic function values, $\operatorname*{Var}_\vttheta \left[ C_\vttheta(\vta) \right] = \operatorname*{\mathbb{E}}_\vttheta \left[ (C_\vttheta(\vta))^2\right] - \left(\operatorname*{\mathbb{E}}_\vttheta \left[ C_\vttheta(\vta) \right]\right)^2$. Regarding $\operatorname*{\mathbb{E}}_\vttheta \left[ C_\vttheta(\vta) \right]$, except for the empty set $J = \emptyset$ (for which $\eta_\emptyset=1$), all the other terms in the summation contain $\sin(2\theta_j)$ for some $j$ and hence vanish in expectation by independence, so $\operatorname*{\mathbb{E}}_\vttheta[C_\vttheta(\vta)]=\kappa^{m}$. Regarding $\operatorname*{\mathbb{E}}_\vttheta \left[ (C_\vttheta(\vta))^2\right]$, for the same reason, $\forall J \neq K$,  $\operatorname*{\mathbb{E}}_{\vttheta}[R_J(\vttheta)R_K(\vttheta)] = 0$, since $R_JR_K$ contains a single factor $\sin(2\theta_j)$ for every $j\in J\triangle K\neq\emptyset$. Putting the remaining terms together gives
\begin{equation*}
     \operatorname*{\mathbb{E}}_\vttheta \left[ C_\vttheta(\vta)^2 \right] = \sum_{J \in \Xi^\vta} \operatorname*{\mathbb{E}}_\vttheta[R_J(\vttheta)^2] = \sum_{J \in \Xi^\vta} \mu^{\abs{J}} \nu^{m-\abs{J}},
\end{equation*}
and subtracting $\operatorname*{\mathbb{E}}_\vttheta[C_\vttheta(\vta)]^2=\kappa^{2m}$ gives the variance.
Then, we derive the variance of the partial derivatives. We differentiate the terms $R_J(\vttheta)$ with respect to $\theta_\ell$. If $\vtsl \notin S^\vta$, then $\comm{Z^{\vtsl}}{X^\vta}=0$, so $C_\vttheta(\vta)$ is independent of $\theta_\ell$, giving $\frac{\partial C_\vttheta(\vta)}{\partial \theta_\ell}=0$. Now, if $\vtsl \in S^\vta$,

\begin{equation*}
\frac{\partial C_\vttheta(\vta)}{\partial\theta_\ell}
=
\sum_{J\in\Xi^\vta} T_J(\vttheta),
\end{equation*}
where
\begin{equation*}
T_J(\vttheta)
=
\begin{cases}
2 \eta_J\prod_{j\in J\setminus\{\ell\}}\sin(2\theta_j)\prod_{j\in (S\setminus J) \cup \{\ell\}}\cos(2\theta_j), & \ell\in J,\\
-2 \eta_J \prod_{j\in J\cup \{\ell \}}\sin(2\theta_j)\prod_{j\in (S\setminus J)\setminus\{\ell\}}\cos(2\theta_j), & \ell\notin J.
\end{cases}
\label{eq:TJ_def}
\end{equation*}
For the same reason as before, $\operatorname*{\mathbb{E}}_{\vttheta}[T_J(\vttheta)T_K(\vttheta)] = 0 $ and  $\operatorname*{\mathbb{E}}_{\vttheta} \left[ \frac{\partial C_\vttheta(\vta)}{\partial\theta_\ell}\right] = 0.$
Therefore,
\begin{equation*}
\operatorname*{Var}_{\vttheta}\left[\frac{\partial C_\vttheta(\vta)}{\partial\theta_\ell}\right]
=
\operatorname*{\mathbb{E}}_{\vttheta}\left[\left(\frac{\partial C_\vttheta(\vta)}{\partial\theta_\ell}\right)^2\right]
=
\sum_{J\in \Xi^\vta}\operatorname*{\mathbb{E}}_{\vttheta}[T_J(\vttheta)^2],
\label{eq:symmetric-derivative-variance}
\end{equation*}
Finally, computing $\operatorname*{\mathbb{E}}_\vttheta[T_J(\vttheta)^2]$ from the two cases of $T_J$ above gives the desired result.
\end{proof}

Now we are ready to prove \cref{thm:uniform}. Consider an $n$-qubit IQP-QCBM in~\cref{def:IQP} with the generator set $S=\{\vtsj\}_{j=1}^D$, and let the reference measure on the parameter space be uniform, namely $\mu_n = \mathcal{U}([0,2\pi)^D)$. For a frequency $\vta \in \mathbb{F}_2^n$, define $S^{\vta} \coloneqq \{ \vtsj\in S \mid \vta \cdot \vtsj = 1 \}$ and the critical rank $r^\vta \coloneqq \dim \operatorname{span}_{\mathbb{F}_2}( S^{\vta})$ (so $r^\vta=0$ if $S^\vta=\emptyset$, in particular $r^{\mathbf 0}=0$). Then
\begin{equation*}
\operatorname*{\mathbb{E}}_{\vttheta\sim\mu_n}\big[C_\vttheta(\vta)^2\big] = 2^{-r^\vta},
\end{equation*}
and if $S^\vta\neq\emptyset$, then $\operatorname*{\mathbb{E}}_{\vttheta}[C_\vttheta(\vta)]=0$, hence $\operatorname*{Var}_{\vttheta\sim\mu_n}[C_\vttheta(\vta)] = 2^{-r^\vta}$. (If $S^\vta=\emptyset$, then $C_\vttheta(\vta)\equiv 1$.) For $\ell \in [D]$, if $\vtsl \notin S^\vta$, then $ \frac{\partial C_\vttheta(\vta)}{\partial \theta_\ell}=0$. If $\vtsl \in S^\vta$, then
\begin{equation*}
\operatorname*{Var}_{\vttheta\sim\mu_n} \left[ \frac{\partial C_\vttheta(\vta)}{\partial \theta_\ell} \right] = 2^{\,2-r^\vta}. \end{equation*}
\begin{proof}
Using the same notations as in the previous proof, in the case of uniform distribution, $\mu = \nu = 1/2$ and $\kappa = 0$. Plugging in these values gives, \begin{equation*}
    \operatorname*{\mathbb{E}}_\vttheta \left[ C_\vttheta(\vta)^2 \right] = 2^{-m^\vta} \abs{\Xi^\vta} , \qquad \operatorname*{Var}_{\vttheta}\left[\frac{\partial C_\vttheta(\vta)}{\partial\theta_\ell}\right]
= 2^{2-m^\vta} \abs{\Xi^\vta} .
\end{equation*}
Note that $\Xi^\vta$ is isomorphic to the set $\{\vtx \in \mathbb{F}_2^{m^\vta} | H \vtx = \bm{0}\}$, where $H$ is the matrix whose columns are the elements of $S^\vta$. In other words, $\Xi^\vta$ is the kernel of $H$ and hence $\abs{\Xi^\vta} = 2^{m^\vta-r^\vta}$ by the rank-nullity theorem (this also holds for $m^\vta=0$, where $\Xi^\vta=\{\emptyset\}$). This gives $\operatorname*{\mathbb{E}}_\vttheta[C_\vttheta(\vta)^2]=2^{-r^\vta}$ and the derivative formula. Finally, $\operatorname*{\mathbb{E}}_\vttheta[C_\vttheta(\vta)]=\kappa^{m^\vta}=0$ whenever $m^\vta\geq1$, i.e.\ $S^\vta\neq\emptyset$, in which case $\operatorname*{Var}_\vttheta[C_\vttheta(\vta)]=\operatorname*{\mathbb{E}}_\vttheta[C_\vttheta(\vta)^2]=2^{-r^\vta}$. This completes the proof.
\end{proof}
 
\subsection{Proof of \cref{lem:parseval}}
Let $Q=\{q_\vttheta\}_\vttheta$ be the ensemble of output distributions of an
$n$-qubit IQP-QCBM with $\vttheta\sim\mathcal U([0,2\pi)^D)$, and let
$\mathcal Z \coloneqq \sum_{\vta\in\mathbb{F}_2^n} 2^{-r^\vta}$, with $r^{\mathbf 0}=0$.
Then the expected collision probability is exactly
\begin{equation}
\operatorname*{\mathbb{E}}_{\vttheta}\Big[\sum_{x\in\{0,1\}^n} q_\vttheta(x)^2\Big]
= \frac{\mathcal Z}{2^n}.
\end{equation}
Consequently, $Q$ is $L^2$ anti-concentrated with constant $b$ if and only if
$\mathcal Z\le b$. 

\begin{proof}
For fixed $\vttheta$, $C_\vttheta(\vta)=\sum_{x}(-1)^{\vta\cdot x}q_\vttheta(x)$ is the
(unnormalised) Walsh--Hadamard transform of $q_\vttheta$. Parseval's identity gives
$\sum_{\vta\in\mathbb{F}_2^n} C_\vttheta(\vta)^2 = 2^n\sum_x q_\vttheta(x)^2$ for every
$\vttheta$. Taking expectations over $\vttheta$ (both sums are finite) and using
$\operatorname*{\mathbb{E}}_\vttheta[C_\vttheta(\vta)^2]=2^{-r^\vta}$ for all $\vta$
(\cref{thm:uniform}) yields \cref{eq:collision-identity}. The equivalence follows by
comparing \cref{eq:collision-identity} with \cref{def:anti-concentration}.
\end{proof}

\subsection{Proof of \cref{thm:erm-sharp}}
\label{app:er-merged}
Let $c>0$ and $p=c\ln n/n$. Let $G\sim G(n,p)$ be an Erd\H{o}s--R\'enyi graph instance. The expected quantity $\mathbb{E}_G[\mathcal Z(G)]$ can be bounded by:
\begin{equation}
2\sum_{k=0}^{K}\frac{\big(T_n(1-o(1))\big)^k}{k!}-\mathcal O\big(n^{2K}2^{-n}\big)\ \le\ \operatorname*{\mathbb{E}}_G\big[\mathcal Z(G)\big]\ \le\ 2\exp\big(T_n(1+o(1))\big)+e^{-\Omega(n/\ln^3n)}
\end{equation}
for every fixed $K\in\mathbb N$, where $T_n=n2^{-n}(2-p)^{n-1}$. In particular, a threshold occurs at $c=2$:
\begin{equation}
    \begin{cases}
    c<2 & \lim\limits_{n\to \infty} T_n = \infty, \lim\limits_{n\to \infty}\mathbb{E}_G[\mathcal Z(G)] = \infty \\
    c=2 & \lim\limits_{n\to \infty}T_n = \frac{1}{2}, \lim\limits_{n\to \infty}\mathbb{E}_G[\mathcal Z(G)] = 2\sqrt{e}\\
    c>2 & \lim\limits_{n\to \infty}T_n = 0, \lim\limits_{n\to \infty}\mathbb{E}_G[\mathcal Z(G)] = 2 
\end{cases}
\end{equation}

\begin{proof}
Recall that by \cref{lem:parseval}, the ensemble of the output distributions of a fixed graph $G$ is anti-concentrated iff $\mathcal Z(G)\in\mathcal O(1)$, where we write $\mathcal Z(G):=\sum_{\vta\in\mathbb F_2^n}2^{-r^\vta(G)}$. 

For a fixed Fourier frequency $\vta \in \mathbb{F}_2^n$, let $A$ be the support of the observable $X^\vta$, and write $\abs{A}=k$. We denote its open neighborhood by $B$. As detailed in \cref{eg:sparse}, $\abs{B} \sim \mathrm{Binomial}(n-k,\,q)$, where we write $q = 1-\lambda^{k}$ and $\lambda=1-p$. Note that $r^\vta = k + \abs{B}$. Evaluating the probability generating function $\mathbb E[x^{|B|}]=(1-q+qx)^{n-k}$ at $x=\frac12$ gives $\mathbb{E}_G[2^{-r^\vta}]=2^{-k}(1-\frac q2)^{n-k}=2^{-k}\big(\frac{1+\lambda^k}{2}\big)^{n-k}$. Summing over all frequencies and binomially expanding $(1+\lambda^k)^{n-k}$ gives
\begin{align}\label{eq:erm-exact}
\operatorname*{\mathbb{E}}_G\big[\mathcal Z(G)\big]
&= \sum_{\vta\in\mathbb F_2^n} \mathbb{E}_G[2^{-r^\vta}] \\
&=\sum_{k=0}^n \binom nk2^{-n} \big(1+\lambda^k\big)^{n-k} \label{eq:t_k sum}\\
&= \sum_{k=0}^n \binom nk2^{-n} \sum_{j=0}^{n-k} \binom {n-k}j \lambda^{jk} \\
&=\sum_{j+k\le n}F(j,k), \label{eq:s_k sum}
\end{align}
where $F(j,k):=\frac{2^{-n}\,n!\,\lambda^{kj}}{k!\,j!\,(n-k-j)!}=F(k,j)$ is symmetric. 

For convenience, let $t_k=\binom nk2^{-k}\big(\frac{1+\lambda^k}{2}\big)^{n-k}$ such that \eqref{eq:t_k sum} writes $\operatorname*{\mathbb{E}}_G\big[\mathcal Z(G)\big] = \sum\limits_{k=0}^n t_k$. Let $s_k:=\sum_{j\ge k}F(k,j)$ and note that $s_0=2^{-n}\sum\limits_{j=0}^n\binom nj=1$. Then following from \eqref{eq:s_k sum}, 
\begin{equation}\label{eq:erm-upper}
\operatorname*{\mathbb{E}}_G\big[\mathcal Z(G)\big]=2\sum_{k=0}^{n} s_k-\sum_{k=0}^{n}F(k,k)\le2+2\sum_{k=1}^{n}s_k .
\end{equation}
Now we prove the upper bound part of \cref{thm:erm-sharp}. Let $k_1:=\lfloor n/\ln^3n\rfloor$ and $k_2:=\lfloor n/\sqrt{\ln n}\rfloor$. We will split the series $\sum\limits_{k=1}^{n} s_k$ into three parts and bound each part separately. 

\medskip
\noindent
For $1\leq k \leq k_1$, it is simpler to bound $t_k$ and then use $s_k\leq t_k$. Using the standard inequalities $\lambda^k=(1-p)^k\le e^{-kp}$, $1-e^{-x}\ge\frac{x}{1+x}$ for $x\ge0$, and $1-x\le e^{-x}$, we have that
\begin{equation}
\frac{1+\lambda^k}{2}\ = 1- \frac{1-\lambda^k}{2} \le 1- \frac{1-e^{-kp}}{2} \le\ 1-\frac{kp}{2(1+kp)}\ \le\ \exp\Big(-\frac{kp}{2(1+kp)}\Big),
\end{equation}
Raising this to the power $n-k$ and using $(n-k)p=c\ln n\,(1-\frac kn)$ gives
\begin{equation}
\Big(\frac{1+\lambda^k}{2}\Big)^{n-k} \leq
\exp\Big(-k\cdot\frac c2\cdot\frac{1-k/n}{1+kp}\,\ln n\Big)=n^{-k\alpha_k}, \quad \alpha_k=\frac c2\cdot\frac{1-k/n}{1+kp}.
\end{equation}

Then, using that $\binom nk\le\frac{n^k}{k!}$ gives
\begin{equation}\label{eq:erm-boundA}
s_k\ \le\ t_k\ \le\ \frac{n^k}{k!}\,2^{-k}\,n^{-k\alpha_k}=\frac1{k!}\Big(\frac12\,n^{1-\alpha_k}\Big)^k .
\end{equation}

We have that $k/n\le1/\ln^3n$ and $kp\le c/\ln^2n$. Using $\frac1{1+x}\ge1-x$ for $x>0$,
\begin{equation}
\alpha_k\ \ge\ \frac c2\Big(1-\frac kn\Big)(1-kp)\ \ge\ \frac c2\Big(1-\frac kn-kp\Big)\ \ge\ \frac c2-\varepsilon_n,
\qquad \varepsilon_n=\frac c2\Big(\frac1{\ln^3n}+\frac c{\ln^2n}\Big).
\end{equation}
Then, on one hand
\begin{equation}
n^{1-\alpha_k}\le n^{1-c/2}\,e^{\varepsilon_n\ln n}=n^{1-c/2}\big(1+o(1)\big).
\label{eq:erm-Sk}
\end{equation}

On the other hand, using $np^2\to0$, $T_n=t_1$ satisfies
\begin{align}
T_n&=n2^{-n}(2-p)^{n-1} \\
&=\frac{n}{2-p}\Bigl(1-\frac{p}{2}\Bigr)^n \\
&=\frac{n}{2-p} \exp\left(-\frac{np}{2} + \mathcal{O}(np^2)\right) \\
&= \Big(\frac{1}{2} + o(1)\Big)n^{1-c/2}
\label{eq:erm-Tn}
\end{align}
Inserting \eqref{eq:erm-Sk} and \eqref{eq:erm-Tn} into \eqref{eq:erm-boundA} gives, for large $n$,
\begin{equation}
\sum_{k=1}^{k_1}s_k\ \le\ \sum_{k=1}^{k_1}\frac{\big(T_n(1+o(1))\big)^k}{k!}\leq\exp\big(T_n(1+o(1))\big)-1. \label{eq:region-1}
\end{equation}

\medskip
\noindent
For $k_1 < k \leq k_2$, \eqref{eq:erm-boundA} still holds, and we will bound $t_k$ with a different method. We have that $k\le n/\sqrt{\ln n}$, so $1-\frac kn\ge1-\frac1{\sqrt{\ln n}}$, and $1+kp\le1+c\sqrt{\ln n}$. Hence, for large $n$,
\begin{equation}
\alpha_k\ \ge\ \frac c2\cdot\frac{1-1/\sqrt{\ln n}}{1+c\sqrt{\ln n}}=\frac{1+o(1)}{2\sqrt{\ln n}}\ \ge\ \frac1{3\sqrt{\ln n}} \implies n^{-\alpha_k}\le e^{-\sqrt{\ln n}/3}.
\end{equation} Moreover, using that $n/k<\ln^3n$ and $k!\ge(k/e)^k$, for large $n$, \eqref{eq:erm-boundA} leads to
\begin{equation}
s_k \leq \Big(\frac e2 \frac nk n^{-\alpha_k}\Big)^k \leq \Big(\frac e2 \ln^3n e^{-\sqrt{\ln n}/3}\Big)^k \leq 2^{-k}
\end{equation}
The last inequality holds for large $n$ because $\frac e2\ln^3n\,e^{-\sqrt{\ln n}/3}\to0$. Hence, for large $n$, \begin{equation}
    \sum_{k=k_1+1}^{k_2}s_k\le2^{-k_1} \label{eq:region-2}
\end{equation}

\medskip
\noindent
For $k>k_2$, we bound $s_k$ directly. Note that
\begin{equation}
s_k=\binom nk2^{-n}\sum_{j=k}^{n-k}\binom{n-k}{j}\lambda^{kj}.
\end{equation}
For $j\geq k$ and $\lambda\leq 1$, we have that $\lambda^{kj}\leq\lambda^{k^2}$. Moreover, $\sum\limits_{j=k}^{n-k}\binom{n-k}{j}\leq \sum\limits_{j=0}^{n-k}\binom{n-k}{j}=2^{n-k}$. Hence,
\begin{equation}
s_k \leq \binom nk2^{-n}\,2^{n-k}\,\lambda^{k^2}=\binom nk2^{-k}\big(\lambda^k\big)^k\ \le\ \frac{n^k}{k!}\,2^{-k}\big(\lambda^k\big)^k=\frac1{k!}\Big(\frac12\,n\lambda^k\Big)^k .
\end{equation}
Using $\lambda^k\le e^{-kp}=e^{-ck\ln n/n}$, we have
\begin{equation*}
s_k\ \le\ \frac1{k!}\Big(\frac12 n e^{-ck\ln n/n}\Big)^k .
\end{equation*}
Using $k > n/\sqrt{\ln n}$, so $e^{-c(k/n)\ln n}<e^{-c\sqrt{\ln n}}$, and $k!\ge(k/e)^k$, for large $n$, we have 
\begin{equation}
s_k \leq \Big(\frac e2\cdot\frac nk\cdot n^{-ck/n}\Big)^k \leq \Big(\frac e2\sqrt{\ln n} e^{-c\sqrt{\ln n}}\Big)^k \leq 2^{-k}.
\end{equation} Hence, for large $n$,
\begin{equation}\label{eq:region-3}
\sum_{k>k_2}s_k\ \le\ 2^{-k_2}.
\end{equation}

Combining \eqref{eq:region-1}, \eqref{eq:region-2}, \eqref{eq:region-3} and \eqref{eq:erm-upper} gives
\begin{equation}
\operatorname*{\mathbb{E}}_G\big[\mathcal Z(G)\big] \leq 2\exp\big(T_n(1+o(1))\big)+2^{1-k_1}+2^{1-k_2}
\leq 2\exp\big(T_n(1+o(1))\big)+e^{-\Omega(n/\ln^3n)},
\end{equation}
which is the upper bound of \cref{thm:erm-sharp}.

Now we prove the lower bound part. For every $0 \leq k \leq n$, the arithmetic-mean-geometric-mean inequality gives $\frac{1+\lambda^k}{2}\ge\lambda^{k/2}$. Hence,
\begin{equation}
 t_k=\binom nk2^{-k}\Big(\frac{1+\lambda^k}{2}\Big)^{n-k} \geq \binom nk2^{-k}(1-p)^{(n-k)k/2}
 \geq \binom nk2^{-k}(1-p)^{nk/2}.
\end{equation}
Since $(1-p)/(1-\frac p2)^2=1-\frac{p^2}{(2-p)^2}$ and $np^2\to0$, we have $(1-p)^{n/2} = (1-\frac p2)^{n}(1-o(1))$. Together with $\binom nk\geq(n-k)^k/k!$, for every fixed $k$ we get
\begin{equation}
 t_k\geq \frac{(n-k)^k}{k!}2^{-k} \Big(\big(1-\tfrac p2\big)^{n}(1-o(1))\Big)^k \geq \frac{\big(T_n(1-o(1))\big)^k}{k!}.
\end{equation}
Now fix a cutoff $K\in\mathbb N$ and sum over $F(j,k)$ for $k\leq K$ or $j \leq K$. By symmetry, they contribute $2\sum\limits_{k=0}^Kt_k$, and their overlap is at most
\begin{equation}
\sum_{j=0}^K\sum_{k=0}^KF(k,j)\leq (K+1)^2n^{2K}2^{-n}.
\end{equation}
Then, we have
\begin{equation}
\operatorname*{\mathbb{E}}_G[\mathcal Z]\geq 2\sum_{k=0}^{K}t_k-\mathcal O(n^{2K}2^{-n})
\geq 2\sum_{k=0}^{K}\frac{\big(T_n(1-o(1))\big)^k}{k!}-\mathcal O(n^{2K}2^{-n}),
\end{equation}
which is the lower bound of \cref{thm:erm-sharp}.
\end{proof}

\subsection{Proof of \cref{cor:erm-typical}}
Let $G\sim G(n,p)$ with $p=c\ln n/n$.

\noindent\textup{(a)} If $c>2$, then with high probability $2-2^{-n}\le\mathcal Z(G)\le2+n^{-(c-2)/4}$.

\noindent\textup{(b)} If $c=2$, then for every fixed $\delta\in(0,1)$, $\Pr_G\big[\mathcal Z(G)\le2+(2\sqrt e-2)/\delta\big]\ge1-\delta-o(1)$.

\noindent\textup{(c)} If $c<2$, then with high probability $\mathcal Z(G)\ge(1-o(1))\,T_n\to\infty$.

\begin{proof}
Every graph satisfies $\mathcal Z(G)\ge2-2^{-n}$, since the $\vta=\bm 0$ term equals $1$ and each of the other $2^n-1$ terms is at least $2^{-n}$. Markov's inequality gives, for every $\varepsilon>0$,
\begin{equation}\label{eq:erm-markov}
\Pr_G\big[\mathcal Z(G)>2+\varepsilon\big] \leq \Pr_G\big[\mathcal Z(G)>2-2^{-n}+\varepsilon\big] \leq \frac{\mathbb{E}_G[\mathcal Z(G)]-2+2^{-n}}{\varepsilon}.
\end{equation}

(a) By \cref{thm:erm-sharp}, $\mathbb{E}_G[\mathcal Z(G)]-2 = \mathcal O(n^{1-c/2})$. Let $\varepsilon=n^{-(c-2)/4}$. Then, $\Pr_G[\mathcal Z(G)>2+\varepsilon]=\mathcal O(n^{-(c-2)/4})\to0$.

(b) By \cref{thm:erm-sharp}, $\mathbb{E}_G[\mathcal Z(G)]\to2\sqrt e$. Let $\varepsilon=(2\sqrt e-2)/\delta$. Then, $\Pr_G[\mathcal Z(G)>2+\varepsilon]\le\delta+o(1)$.

(c) For simplicity, we will show that the sum of all weight-$1$ frequency contributions in $\mathcal Z(G)$ alone already diverges to positive infinity with high probability. Let \begin{equation}
X(G)=\sum\limits_{v=1}^{n}Y_v, \qquad Y_v=2^{-r^{\bm e_v}(G)}=2^{-1-\deg(v)},
\end{equation} since $S^{\bm e_v}$ consists of $\bm e_v$ and the edges at $v$. Let $\beta=1-\frac p2$. Using $\deg(v)\sim\mathrm{Bin}(n-1,p)$ gives $\mathbb{E}_G[X]=\frac n2\beta^{n-1}=T_n$. Now we compute the variance terms. $\mathbb{E}_G[Y_v^2]=\frac14(1-\frac{3p}4)^{n-1}\le\frac12\mathbb{E}_G[Y_v]$. For $u\ne v$, splitting off the shared edge $(u,v)$ gives $\mathbb{E}_G[Y_uY_v]=\frac14(1-\frac{3p}4)\beta^{2n-4}$, so $\operatorname{Cov}(Y_u,Y_v)=\frac14\beta^{2n-4}\cdot\frac p4(1-p)\le p\,(\mathbb{E}_G[Y_v])^2$. Hence, $\operatorname{Var}_G[X]\le\frac{T_n}2+p\,T_n^2$. Chebyshev's inequality gives, for every $\eta\in(0,1)$,
\begin{equation}
\Pr_G\big[X(G)\le(1-\eta)T_n\big]\le\frac1{\eta^2}\Big(\frac1{2T_n}+p\Big).
\end{equation}
By \eqref{eq:erm-Tn}, $T_n\to\infty$ for $c<2$, and $p\to0$. Choosing $\eta=(1/T_n+p)^{1/3}\to0$ yields $\mathcal Z(G)\ge X(G)\ge(1-o(1))\,T_n$ with high probability.
\end{proof}

\subsection{Proof of~\cref{lemma:low-rank-tail}}
\label{app:low-rank-tail}
Fix $k_0 \geq 1$ independent of $n$ and a sufficiently large constant
$c>2$. For $G\sim G(n,c\ln n/n)$, we define $R_{k_0}(G) \coloneqq \max_{1\leq\abs{\vta}\leq k_0}r^{\vta}(G)$.
There exist constants $K,\gamma>0$ such that $E_n \coloneqq \{G \colon R_{k_0}(G)\leq K\ln n\}$ satisfies $\delta_n \coloneqq\Pr_{G\sim G(n,c\ln n/n)}[G\notin E_n] \in \mathcal{O}(n^{-\gamma})$.
\begin{proof}
Let $p=c\ln n/n$. By the binomial law in Example~\ref{eg:sparse}, for $k=\abs{\vta}\leq k_0$ we have $r^{\vta}=k+X_{\vta}$, where $X_{\vta}\sim\operatorname{Bin}(n-k,1-(1-p)^k)$ and $\mathbb{E} X_{\vta}\leq c k_0\ln n$. Choose $L>ec k_0$ so that $\rho:=L\ln(L/(ec k_0))>k_0$, and set $K=L+k_0$ and $\gamma=\rho-k_0>0$. For $n\geq e$, a Chernoff bound and a union bound yield
\begin{align}
 \Pr_{G\sim G(n,p)}\left(R_{k_0}(G)>K\ln n\right)
 &\leq\sum_{1\leq\abs{\vta}\leq k_0}\Pr\left( X_{\vta}>L\ln n \right) \nonumber\\
 &\leq k_0 n^{k_0}\left(\frac{ec k_0}{L}\right)^{L\ln n} = k_0n^{k_0-\rho} = \mathcal{O}(n^{-\gamma}).
 \label{eq:low-rank-rank-tail}
\end{align}
No independence between different frequencies is required.
\end{proof}

\subsection{Proof of~\cref{thm:sampling-hardness}}
\label{app:sampling-hardness}
Assume Conjecture~8 of \cite{Bremner17} and noncollapse of the polynomial
hierarchy. Fix $k_0\in\mathcal{O}(1)$ and choose $K$ as above. There exist constants $\varepsilon>0$ and $\eta>0$ such that
\begin{enumerate}
    \item There is no uniform polynomial-time classical sampler $A$ which takes a description of circuit $C$ and outputs a distribution $Q_C^A$ satisfying
    \begin{equation} \label{eq:low-rank-average-hardness}
    \Pr_{C\sim\mathcal D_n^{\rm low}}\left(\|Q_C^A-P_C\|_1\leq\varepsilon \right) \geq 1 - \eta
    \end{equation}
    for all sufficiently large $n$.
    \item There is no uniform polynomial-time classical sampler $A$ whose output distribution $Q_C^A$ satisfies
    \begin{equation}
    \|Q_C^A-P_C\|_1\leq\varepsilon \quad\text{for every } C\in\mathcal{C}_n^{\rm low}
    \label{eq:low-rank-sampling-error}
    \end{equation}
    for all sufficiently large $n$.
\end{enumerate}
\begin{proof}
We first recap the definition of the distribution $\mathcal{D}_n$ in the original work.

\paragraph{The original sparse IQP distribution in~\cite[Section~2]{Bremner17}.}
An instance of sparse IQP circuit is generated by first sampling $G\sim G(n,c\ln n/n)$ and then choosing, independently conditional on $G$,
\begin{equation} \label{eq:diagonal-circuit-BMS17}
    D_C=\prod_{\{i,j\}\in E(G)}\operatorname{diag}(1,1,1,i^{w_{ij}}) \prod_{i=1}^n\operatorname{diag}(1,\zeta^{v_i}), \quad \zeta=e^{\pi i/4}, 
\end{equation}
where $w_{ij}$ is uniform on $\{0,1,2,3\}$ and $v_i$ is uniform on $\{0,\ldots,7\}$. Its partition function satisfies
\begin{equation}\label{eq:partition-function}
    Z_C=\sum_{z\in\{0,1\}^n} \zeta^{\,2\sum_{\{i,j\}\in E(G)}w_{ij}z_iz_j+\sum_i v_iz_i}, \quad P_C(0)=4^{-n}|Z_C|^2.
\end{equation}
Conjecture~8 of \cite{Bremner17} postulates that approximating $|Z_C|^2$ to relative error $1/4+o(1)$ on an arbitrary $c$ fraction of this original joint graph-and-label ensemble is \#P-hard, for a sufficiently small fixed constant $c>0$ specified by that conjecture.

Next, we show our reformulation of $\mathcal{D}_n$ is equivalent. By the following identities
\begin{align*}
 \operatorname{diag}(1,e^{it})
 &=e^{it/2}e^{-itZ_i/2},\\
 \operatorname{diag}(1,1,1,e^{i\phi})
 &=e^{i\phi/4}e^{-i\phi Z_i/4}e^{-i\phi Z_j/4}e^{i\phi Z_iZ_j/4}
\end{align*}
we see that the diagonal part of each original instance agrees, up to a global phase, with $U(G,\vttheta)$ generated from $\mathcal{D}_n$ in the main text.

We then recap the sampling-hardness theorem from~\cite{Bremner17}, which is used in this proof. 
    
\begin{theorem}[Rewrite of {\cite[Theorem~3]{Bremner17}}] \label{thm:BMS17}
Assume Conjecture~8 of \cite{Bremner17} and noncollapse of the polynomial
hierarchy. There exist constants $\varepsilon_*>0$ and $q_*\in(0,1)$
such that no uniform polynomial-time classical sampler is
$\varepsilon_*$-accurate in $\ell_1$ distance on a $1-q_*$ fraction of
the original $\mathcal D_n$ ensemble for all sufficiently large $n$.
\end{theorem}
The proof relies on the anti-concentration bound of ensemble $\mathcal{D}_n$ (see~\cite[Lemma~6]{Bremner17}) and the
Stockmeyer reduction (see~\cite[Corollary~7]{Bremner17}).
For the anti-concentration bound, one may fix $\alpha,\beta>0$ such that
\begin{equation} \label{eq:discrete-anticoncentration}
 \Pr_{C\sim\mathcal D_n}[P_C(0)\geq\alpha 2^{-n}]\geq\beta,
 \qquad c < \beta/2.
\end{equation}
To use~\cref{thm:BMS17}, we need to prove the anti-concentration bound for low-rank ensemble $\mathcal{D}_n^{\rm low}$.

\paragraph{Anti-concentration of $\mathcal{D}_n^{\rm low}$.}
Denote circuit instance by $C=(G,\theta)$. Since $\Pr_{C\sim\mathcal{D}_n}[G\notin E_n]=\delta_n=o(1)$, we have
\begin{align} 
 \Pr_{C\sim\mathcal{D}_n^{\rm low}}\left(P_C(0)\geq\alpha2^{-n}\right) =\frac{\Pr_{C\sim\mathcal D_n}\left(P_C(0)\geq\alpha2^{-n},\ G_C\in E_n\right)}{1-\delta_n} \geq \frac{\beta-\delta_n}{1-\delta_n}=\beta-o(1).
\end{align}
This proves anti-concentration of $\mathcal D_n^{\rm low}$ without requiring independence between the graph restriction and output probabilities.

\paragraph{Average-case hardness for the low-rank family $\mathcal{D}_n^{\rm low}$.}
Take $\varepsilon=\varepsilon_*$ and fix $0<\eta<q_*$.
Suppose a uniform polynomial-time sampler $A$ satisfies \cref{eq:low-rank-average-hardness}. Define an extension on the original
discrete ensemble by
$$
 A'(C)=
 \begin{cases}
  A(C), &G_C\in E_n,\\
  \text{a uniform bit string in }\{0,1\}^n, &G_C \notin E_n.
 \end{cases}
$$
This is polynomial-time because for fixed $k_0$, membership in $E_n$ is polynomial-time decidable by
computing $r^{\vta}(G)$ for the $\mathcal{O}(n^{k_0})$ relevant frequencies.
Write $Q_C^A$ for its output law.
On $E_n$, we have $Q_C^{A'}=Q_C^A$, so the law of total probability gives
$$
\Pr_{C\sim\mathcal{D}_n}\left(\|Q_C^{A'}-P_C\|_1>\varepsilon_*\right) \leq (1-\delta_n)\Pr_{C\sim\mathcal{D}_n^{\rm low}}\left(\|Q_C^A-P_C\|_1 > \varepsilon\right) + \delta_n \leq(1-\delta_n)\eta+\delta_n < q_*
$$
for sufficiently large $n$, because $\delta_n\to0$ and $\eta<q_*$. Thus $A'$ is $\varepsilon_*$-accurate on more than a $1-q_*$ fraction
of the original $\mathcal{D}_n$ ensemble, contradicting~\cref{thm:BMS17}. For fixed $\eta$, the tail bound tends to $\eta$; it suffices to choose this constant below $q_*$. 

\paragraph{Details of the average-case hardness argument.} Here, we aim to present some details behind the above contradiction to show where the anti-concentration bound is needed precisely. First, we write $Q_C \coloneqq Q_C^{A'}$ and define its mean sampling error under $\mathcal{D}_n$
\begin{equation} \label{eq:mean-sampling-error}
    b_n \coloneqq \mathbb{E}_{C\sim\mathcal{D}_n}\|Q_C-P_C\|_1 \leq(1-\delta_n)(\varepsilon+2\eta)+2\delta_n \leq\varepsilon+2\eta+2\delta_n\,.
\end{equation}
Next, draw $C\sim\mathcal D_n$ and $x\sim\mathcal{U}(\{0,1\}^n)$
independently, and let $C^{(x)}=X^xC$, where $X^x=\bigotimes_iX_i^{x_i}$ flips the selected output bits. Then
\begin{equation}
 P_{C^{(x)}}(y)=P_C(y\oplus x),\qquad
 P_{C^{(x)}}(x)=P_C(0).
 \label{eq:appendix-output-hiding}
\end{equation}
The shift preserves both $\mathcal{D}_n$ and the event $E_n$. For each fixed $x$, $C^{(x)}$ therefore still follows the law $\mathcal{D}_n$; in particular, $(C^{(x)},x)$ has the product law. Hence the average error at this
hidden output is
\begin{equation}
    \mathbb{E}_{C,x}|Q_{C^{(x)}}(x)-P_C(0)| =2^{-n}\mathbb E_{C'\sim\mathcal D_n} \sum_y|Q_{C'}(y)-P_{C'}(y)| =b_n2^{-n}.
\end{equation}
Markov's inequality consequently gives
\begin{equation}\label{eq:hidden-error-tail}
    \Pr_{C,x}\left(|Q_{C^{(x)}}(x)-P_C(0)|>(\alpha/4)2^{-n}\right) \leq\frac{4b_n}{\alpha}.
\end{equation}

Stockmeyer's approximate-counting theorem supplies an $\mathrm{FBPP}^{\mathrm{NP}}$ estimate $\widetilde Q$ with
$|\widetilde Q-Q_{C^{(x)}}(x)|\leq\tau Q_{C^{(x)}}(x)$, for inverse-polynomial $\tau>0$. On the anti-concentration event $P_C(0)\geq\alpha2^{-n}$ from \cref{eq:discrete-anticoncentration} and outside the error event in \cref{eq:hidden-error-tail}, we have
$$
|Q_{C^{(x)}}(x)-P_C(0)| \leq(\alpha/4)2^{-n} \leq P_C(0)/4
$$
and $Q_{C^{(x)}}(x) \leq 5P_C(0)/4$.
This is precisely where anti-concentration is needed: it converts an additive error of order $2^{-n}$ into a relative error. The triangle inequality now yields
\begin{align}
        |\widetilde Q_{C^{(x)}}(x)-P_C(0)| \leq&|\widetilde{Q}-Q_{C^{(x)}}(x)|+|Q_{C^{(x)}}(x)-P_C(0)| \\ \leq&\tau Q_{C^{(x)}}(x) + P_C(0)/4 \\ \leq& \left(\tfrac14+o(1)\right)P_C(0) \,.
\end{align}

By a union bound, this estimate succeeds with probability at least $\beta-4b_n/\alpha-o(1)$ over the instance $C$. Thus $\varepsilon$ and $\eta$ can be chosen sufficiently small that, using \cref{eq:mean-sampling-error} and $\delta_n=o(1)$, the success bound exceeds the fixed fraction stipulated in Conjecture~8 of \cite{Bremner17}, as in their original argument.

\paragraph{Worst-case hardness for the low-rank ensemble $\mathcal{C}_n^{\rm low}$.}
A sampler $A$ satisfying \cref{eq:low-rank-sampling-error} for every circuit with $G\in E_n$ would, in particular, sample every instance under $\mathcal{D}_n^{\rm low}$, since the low-rank circuits with discrete angles described in the main text form a subset of the circuits with graphs in $E_n$. The sampler $A$ would therefore succeed with probability one under $\mathcal{D}_n^{\rm low}$, contradicting the average-case hardness proved above.
\end{proof}

\section{Setup for numerical validation}\label{sec:setup}
To train the models, we use MMD loss with a kernel that consists of weight-2 and weight-1 frequencies. We train all models with the Adam optimizer, learning rate 0.02 (obtained from hyperparameter search), $\beta_1=0.9, \beta_2=0.999$ (by default), where the initial parameters are sampled u.a.r.~in $[0, 2\pi)^D$. For each architecture, we train for 5 resource budgets (sample sizes from 64 to 1024) using the MMD loss. If the sample size is too small, the Monte Carlo noise will be too large to be overcome by the gradient signal, thereby hindering training. The sample size cannot be too large, as this would increase training cost and cause an out-of-memory problem.
To demonstrate the full training of IQP-QCBM on real-world datasets, we also perform the training on the binarized MNIST dataset (requires 784 qubits). We test three architectures: 2D Lattice, Sparse ER, and Complete graph. The training is conducted on an RTX 6000 ADA GPU.

\section{Architecture Illustration}
\vspace{-0.5cm}
\begin{figure}[h]
  \begin{center}    \centerline{\includegraphics[width=0.5\columnwidth]{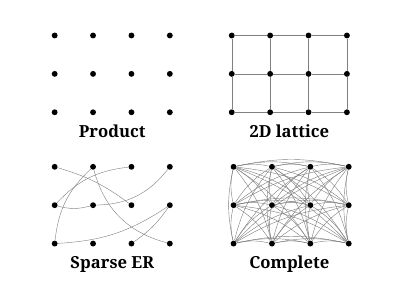}}
  \vspace{-0.5cm}
    \caption{Illustration of the IQP architectures considered in this paper: product state (\cref{eg:product}), 2D lattice (\cref{eg:lattice}), sparse Erd\H os-R\'enyi graph (\cref{eg:sparse}), and complete graph (\cref{eg:complete}).}
    \label{pic:illustration}
  \end{center}
\end{figure}

\end{document}